\documentclass[12pt,twoside]{article}
\usepackage[mathscr]{eucal}
\usepackage{amsmath,amsfonts,amssymb,amsthm}

\usepackage[matrix,arrow]{xy}
\usepackage{times}
\usepackage{graphicx}% Include figure files
\usepackage{epstopdf}
\usepackage{hyperref}

\advance\textheight\topskip
\renewcommand{\tilde}{\widetilde}

\newtheorem{prop}{Proposition}[section]

\newtheorem{definition}[prop]{Definition}
\newtheorem{theorem}[prop]{Theorem}
\newtheorem{corollary}[prop]{Corollary}
\newtheorem{rem}{Remark}[section]
\newtheorem{example}[rem]{Example}

\renewcommand{\d}{\partial}

\newcommand{\RR}{\mathbb{R}}

\def\cC{\mathcal{C}}

\def\cF{\mathcal{F}}

\def\cK{\mathcal{K}}
\def\cL{\mathcal{L}}
\def\cM{\mathcal{M}}

\def\cP{\mathcal{P}}

\def\cR{\mathcal{R}}

\def\cW{\mathcal{W}}

\numberwithin{equation}{section} \makeatletter
\@addtoreset{equation}{section}

\begin{document}

\def\mytitle{Canonical Structure of Field Theories with Boundaries \\
and Applications to Gauge Theories}

\pagestyle{myheadings} \markboth{\textsc{\small Troessaert}}{%
  \textsc{\small Hamiltonian Field Theories with Boundaries}} \addtolength{\headsep}{4pt}

\begin{centering}

  \vspace{1cm}

  \textbf{\Large{\mytitle}}

%\vspace{1cm}

%{\huge Notes}

  \vspace{1.5cm}

  {\large C\'edric Troessaert$^a$}

\vspace{.5cm}

\begin{minipage}{.9\textwidth}\small \it \begin{center}
   Centro de Estudios Cient\'ificos (CECs)\\ Arturo Prat 514,
   Valdivia, Chile \\ troessaert@cecs.cl \end{center}
\end{minipage}

\end{centering}

\vspace{1cm}

\begin{center}
  \begin{minipage}{.9\textwidth}
    \textsc{Abstract}. In this paper, we present a review of the
    canonical structure of field theories defined on manifolds with
    time-like boundaries. The notion of differentiable generator is
    shown to be a requirement coming from the consistency of the
    symplectic structure. We show how this structure can be
    applied to classify the possible boundary conditions of a general gauge
    theory. We then review the definition and properties of surface
    charges. We show how the notion of
    differentiable generators allows the direct computation of the
    phase-space of boundary gauge degrees of freedom.
  \end{minipage}
\end{center}

\vfill

\noindent
\mbox{}
% \raisebox{-3\baselineskip}{%
%   \parbox{\textwidth}{\mbox{}\hrulefill\\[-4pt]}}
{\scriptsize$^a$Laurent Houart postdoctoral fellow.}

\thispagestyle{empty}
\newpage

\begin{small}
{\addtolength{\parskip}{-1.5pt}
 \tableofcontents}
\end{small}
\newpage

\section{Introduction}
\label{sec:introduction}

The usual way one develops the lagrangian and hamiltonian formulation for field
theories is by taking the continuous limit from a discrete system 
\cite{Goldstein2001}. This allows for a local definition of
field theories. For most applications we can neglect boundary
contributions and this structure is sufficient: either we work on manifolds without boundary or we impose
asymptotic conditions strong enough to put all boundary contributions
to zero. Unfortunately, those setups are too restrictive for a lot of
physically relevant problems. In general, we need to relax the
asymptotic behavior and deal with the boundary terms. 

\vspace{5mm}

The situation is easily solved in the Lagrangian framework: one
chooses boundary conditions and adds a corresponding boundary term to
the Lagrangian in order to make it well-defined. A well-defined
Lagrangian $L$ is such that, for any allowed variation of
the fields $\delta \phi^a$, the variation of $L$ does not produce any
boundary term: 
\begin{equation}
\delta \int_\cM \, L \, d^{n+1}x = \int_\cM \, \frac{\delta L}{\delta \phi^a}
\delta \phi^a \, d^{n+1}x .
\end{equation}
This well-defined action is the one needed in the path integral
\cite{Gibbons1977}.

In the hamiltonian picture, it seems that we don't need any
modification as the fundamental object, the poisson bracket, is
independent of total derivatives. In the same way that
hamiltonian generators are defined up to a constant in discrete
classical mechanics, they are usually defined up to a boundary term in field
theories. Unfortunately, this setup does not always work. 

A famous problem in that context is the definition of mass in general
relativity. The main issue is that the hamiltonian density is given by
a sum of constraints and is zero on all solutions of the equations of
motion. The answer is to add to the Hamiltonian a well chosen boundary
term \cite{Arnowitt2008}: this doesn't change the equations of motion and gives the expected
value for the energy. However, this construction is ad-hoc and is not a
solution to the problem as the formalism does not constrain this
boundary term. 
This phenomenon is general and appears whenever one wants to define conserved quantities
related to gauge symmetries. The generators of gauge symmetries are
constraints and always give zero when evaluated on solutions. The
conserved quantities are then given by specific boundary terms known as
surface charges: one
example is the electric charge given by the flux of the electric field
through the boundary. As for gravity, these boundary terms are ad-hoc
and not constrained by the formalism.

The final solution was proposed by Regge-Teitelboim in
\cite{Regge1974} and was later refined in
\cite{Benguria1977,Henneaux1985,Brown1986,Brown1986a}. Using
the definition of well-defined action explained above and applying it
to the hamiltonian action fix the form of the boundary
term in the definition of the Hamiltonian. This leads to the notion of
differentiable generator. Restricting the set of functionals to the set
of differentiable functionals fixes the boundary term for any hamiltonian
generator. In particular it fixes the boundary term for gauge
symmetries and gives a systematic definition for the algebra of
asymptotic symmetries and their associated surface charges. 

\vspace{5mm}

In the last 15 years, this technique and its lagrangian counterpart
\cite{Barnich2002,Barnich2003,Barnich2008} have been very useful in the
study of holography and of the $AdS/CFT$ conjecture. The conjecture relates a bulk
gravity theory to a field theory without gravity living in one
dimension lower. The two theories being equivalent, they share the same
symmetry algebra. In particular, the asymptotic symmetries of the
bulk gravity theory are global symmetries of the dual theory. The most
famous exemple in that context is maybe the original computation by
Brown-Henneaux in \cite{Brown1986}. They showed that the asymptotic
symmetry algebra of the asymptotically $AdS_3$ space-times is given by
two copies of the Virasoro algebra with central charges given by
$c^\pm=\frac{3l}{2G}$. This proves that this theory is described by a
conformal field theory in two dimensions. Recently, it has played a major role in the study of higher spin
theories in 3 dimensions and their holographic duals. In \cite{Campoleoni2010,Henneaux2010,Gaberdiel2011}, it was
shown that higher spin theories have asymptotic symmetry algebras
given by $\cW$-algebras. This provides good indications in favor of
the conjecture that higher spin theories on $AdS_3$ are duals to
certains minimal cosets models (see \cite{Gaberdiel2012} for a review).

\vspace{5mm}

In this paper, we present a constructive introduction to the
notion of differentiable functional. We show how this structure can
be applied to classify the possible boundary conditions of a general
gauge theory. We then review the definition and properties of surface
charges associated to asymptotic gauge symmetries. In the last part of
the paper, we show how the same notion of 
differentiable generator allows the direct computation of the reduced
phase-space of some topological theories. The appendix contains the
definitions and conventions used to describe the differential
structure of the phase-space of field theories.

The plan of the paper is the following:
\begin{itemize}
\item In section \ref{sec:CF-FieldTheories}, we present the canonical
  structure of field theories defined on a manifold with
  boundary. Requiring field theories to behave like discrete
  mechanical systems naturally introduces the notion of
  differentiable functional. We then show how this structure is
  related to the notion of well-defined action and we end with a review
  of Noether's theorem.

\item Section \ref{sec:surface-charges} is devoted to gauge
  theories. We start by describing the set-up and the requirements for
  consistent boundary conditions. We then introduce the notion of
  differentiable gauge generator and use them to classify the
  possible boundary conditions on the lagrange multipliers. The last
  part contains a review of the definition of the asymptotic symmetry
  algebra and the associated surface charges.

\item In section \ref{sec:reduced-phase-space}, we show how
  the notion of differentiable gauge generator allows the computation
  of the phase-space of boundary gauge degrees of freedom  of some
  topological theories without the 
  need to solve the constraints. We also
  present how one can make a complete classification of the possible
  boundary conditions. The technique is presented using Chern-Simons
  in three dimensions and BF theory in four dimensions as examples. 
\end{itemize}

%%%%%%%%%%%%%%%%%%%%%%%%%%%%%%%%%%%%%%%%%%%%%%%%%%%%%%%%%%%%%%%%%%%%%%
\newpage

\section{Canonical Structure for Field Theories}
\label{sec:CF-FieldTheories}

We explained in the introduction that one has to be careful with the
boundary terms when studying field theories on manifolds with boundaries. The
problem is even deeper: in presence of
time-like boundaries, the usual Poisson bracket does not
satisfy Jacobi's identity. It means that the canonical structure is
not well-defined.

Let's consider a simple example: 3D
Chern-Simons theory on a cylinder $\RR \times  D$ with standard
coordinates $x^\mu= (t, r, \phi)$ the time-like boundary being given
by $r=R$. The action is given by 
\begin{equation}
S[A^a_\mu] = \frac{-\kappa}{2\pi} \int dt \int_D d^2x \,
\frac{1}{2} \epsilon^{ij}g_{ab} \left(A^a_i \dot A^b_j - A^a_0 F^b_{ij} \right),
\end{equation}
where $\epsilon^{12}=1$ and the metric $g_{ab}$ is a symmetric
non-degenerate invariant tensor on the Lie algebra $\mathfrak g$. We
use the field strength $F^a_{ij}=\d_i A^a_j - \d_j A^a_i +
f^a_{bc}A^b_iA^c_j$ with $f^a_{bc}$ being the structure constants of $\mathfrak g$.
If we impose the boundary condition $A^a_0\vert_{\d D}=0$, the action
is well-defined and the lagrangian picture makes sense \cite{Moore1989,Elitzur1989}. Let's now
compute Jacobi's identity for the following gauge generators: 
\begin{gather}
I = \frac{-\kappa}{4\pi} \int_D d^2x \, \rho A^a_r
\epsilon^{ij}g_{ab}F^b_{ij}, \qquad  J =\frac{-\kappa}{4\pi} \int_D
d^2x \, \eta^a \epsilon^{ij}g_{ab}F^b_{ij},\\ K=\frac{-\kappa}{4\pi} \int_D d^2x \, \xi^a \epsilon^{ij}g_{ab}F^b_{ij},
\end{gather}
where $\eta^a, \xi^a, \rho$ are independent of the dynamical fields
and $\rho=0$ in a neighborhood of the origin. For simplicity, let's assume that
$\eta^a\vert_{\partial D}=0$ and $\rho=1$ on a neighborhood of
$\partial D$. A straightforward computation gives:
\begin{equation}
\left\{ I, \left\{ J, K\right\}\right\} + \left\{ J, \left\{ K,
    I\right\}\right\} + \left\{ K, \left\{ I, J\right\}\right\} \approx
\frac{\kappa}{2\pi} \oint_{\partial D} d\phi \, \partial_r \eta^a
g_{ab} D_\phi \xi^b.   
\end{equation}
which is non-zero in general. The covariant derivative is defined by
$D_i \xi^a=\d_i \xi^a + f^a_{bc}A^b_i\xi^c$ and we used the symbol $\approx$ to denote equality on
the constraints surface.

The notion of differentiable functional introduced by Regge and
Teitelboim in \cite{Regge1974} solves this problem and allows for a
good definition of the canonical structure. The idea is to restrict the set of allowed functionals to the set of
differentiable functionals. In the above example, the functional $I$
is not differentiable and should not be allowed in the Poisson
bracket. More general definitions of the canonical
structure in presence of a boundary have been developed in
\cite{Lewis1986,Soloviev1993,Soloviev1995}, but they add non-trivial dynamics on the
boundary and will not be needed for our description.

\vspace{5mm}

In this section, we will present a constructive introduction to
the Regge-Teitelboim idea and its link to the Lagrangian framework. We
will start with the description of the symplectic structure for field
theories and introduce the idea of differentiable functionals. We will
then make the link with the lagrangian notion of well-defined
action. In the last part, we will show that Noether's theorem
associates a differentiable generator to any symmetry of the action.

The main point of this construction is that, using the
notion of differentiable generators, 
the hamiltonian structure of field theories behaves exactly as the one of
discrete mechanical systems.

\subsection{Symplectic Structure and Poisson Bracket}
\label{sec:Symplectic-struct}

The notion of symplectic manifold can be taken as the starting point of the hamiltonian theory
\cite{Arnold1989}:
\begin{definition}
Let $\cM$ be an even-dimensional differentiable manifold. A
symplectic structure on $\cM$ is a closed non-degenerate
differential 2-form $\omega$ on $\cM$:
\begin{equation}
d \omega = 0 \qquad \text{and} \quad \forall \xi \ne 0, \exists \eta:
\omega(\xi, \eta) \ne 0, \, \xi, \eta \in T_x\cM.
\end{equation}
The pair $(\cM, \omega)$ is called a symplectic manifold.
\end{definition}
\noindent If we have Darboux
coordinates on $\cM$, the symplectic structure takes the form:
\begin{equation}
\omega = dp_i dq^i
\end{equation}
where the $q^i$ describe the position of the system and the $p_i$ are the
associated momenta.

\vspace{5mm}

For field theories, the equivalent of $\cM$ is the set of allowed
configurations of the fields $z^A$ that we will denote $\cF(\Sigma) =
\left\{z^A(x), x^i\in \Sigma; \chi^\mu(z)\vert_{\d \Sigma} = 0
\right\}$. The manifold $\Sigma$ describes constant time
slices of the space-time under consideration. The conditions $\chi^\mu(z)\vert_{\d \Sigma} = 0$
are the set of boundary condition. We will 
assume in the following that they
are imposed in all equalities.

A differential 2-form on $\cF(\Sigma)$ is a functional 2-form
(see appendix~\ref{sec:phase-space}). We will restrict ourselves to the
most simple case where we have Darboux coordinates for the fields and
we will assume that the symplectic structure is given by:
\begin{equation}
\label{eq:SymplStruct}
\Omega = \int_\Sigma \frac{1}{2} \sigma_{AB} \delta z^A \delta z^B d^nx,
\end{equation}
where $\sigma_{AB}$ is a non degenerate constant antisymmetric
matrix whose inverse will be denoted $\sigma^{AB}$. The results we
will present in this paper apply only to this 
case. 
\begin{example} 
$\quad$

\begin{itemize}
\item Electromagnetism: the phase-space can be
  parametrized by
\begin{equation}
z^A = (A_i, E^i)
\end{equation}
where $A_i$ is the potential vector and $E^i$ is the electric
field. The symplectic structure is then given by
\begin{equation}
\Omega = \int_\Sigma -\delta E^i \delta A_i \, d^nx
\end{equation}
\item Gravity: in this case, we can use the spacial metric $g_{ij}$
  and its conjugate momentum $\pi^{ij}$ to describe the phase-space:
\begin{equation}
z^A = (g_{ij}, \pi^{ij})
\end{equation}
and
\begin{equation}
\Omega = \int_\Sigma \delta \pi^{ij} \delta g_{ij} \, d^nx.
\end{equation}
\end{itemize}
\end{example}

\vspace{5mm}

The symplectic structure of a manifold defines an isomorphism between
1-forms and vector fields. In the field-theoretic case, the 1-forms are
functional 1-forms and the vector fields are evolutionary vector
fields preserving the boundary conditions. From an allowed
evolutionary vector field $Q^A \frac{\d}{\d z^A}$ we can build a
functional 1-form $\Theta_Q$ using the symplectic structure
\eqref{eq:SymplStruct}:
\begin{equation}
\label{eq:vect2form}
\Theta_Q = \iota_Q \Omega = \int_\Sigma \sigma_{AB} Q^A \delta z^B d^nx.
\end{equation}
Due to the particular form of $\Theta_Q$ and the restrictions on
$Q^A$, the image of this application is never the full set of functional
1-forms.

\begin{definition}
A differential 1-form is a functional 1-form of the form
\begin{equation}
\Theta  = \int_\Sigma \theta_A \, \delta z^A d^nx,
\end{equation}
such that the evolutionary vector field $\sigma^{AB}\theta_B \frac{\d}{\d z^A}$ preserves
the boundary conditions.
\end{definition}
\noindent Any functional 1-form can be put into this form up to boundary terms. The key point
is that these boundary terms must be zero using the boundary
conditions. The application \eqref{eq:vect2form} defines an
isomorphism between the differential 1-forms and the evolutionary
vector fields preserving the boundary conditions. We will denote by
$J$ the inverse of this isomorphism.  

\vspace{5mm}

The fact that we needed to restrict the set of functional 1-forms in
order to have an isomorphism with the evolutionary vector fields means
that we will not be able to associate a hamiltonian vector field to
all functionals. Only functionals for which the differential $\delta$
gives a differential 1-form will generate a hamiltonian
transformation. This leads to the following definition:
\begin{definition}
A functional $G =\int_\Sigma g \,d^nx$ is called differentiable if its 
differential $\delta G$ is a differential 1-form:
\begin{equation}
\label{eq:propdiff}
\delta G = \int_\Sigma \, \frac{\delta g}{\delta z^A} \delta z^A d^nx \quad
\Leftrightarrow \quad \oint_{\d \Sigma} I^n ( g d^nx)=0
\end{equation}
and the evolutionary vector field $\sigma^{AB}\frac{\delta g}{\delta z^B}\frac{\d}{\d z^A}$ preserves
the boundary conditions.
\end{definition}
\noindent The
property \eqref{eq:propdiff} can also be written in term of
evolutionary vector fields by
asking that for all variations $\delta_Q$ preserving the boundary
conditions, we have
\begin{equation}
\label{eq:propdiffII}
\delta_Q G = \int_\Sigma \, \frac{\delta g}{\delta z^A} Q^A d^nx \quad
\Leftrightarrow \quad \oint_{\d \Sigma} I^n_{Q} ( g d^nx)=0.
\end{equation} 
This definition of differentiable functional is the one
introduced in \cite{Regge1974}, but we see that 
it comes naturally from the analysis of the symplectic structure. 
\begin{definition}
An evolutionary vector field $Q^A \frac{\d}{\d z^A}$ is called
hamiltonian if there exists a differentiable functional $G =
\int_\Sigma g \,d^nx$ such that
\begin{equation}
Q^A \frac{\d}{\d z^A} = J \delta G \qquad \Leftrightarrow \qquad Q^A =
\sigma^{AB} \frac{\delta g}{\delta z^B}.
\end{equation}
The functional $G$ is the generator of $Q^A \frac{\d}{\d z^A}$.
\end{definition}

\vspace{5mm}

Using these definitions, field theories behave in exactly the same
way as standard mechanical systems. We will now derive some of the
most important hamiltonian results that we will need later.
\begin{prop}
\label{pr:uptoaconst}
Let the phase-space $\cF(\Sigma)$ be path-connected. If two differentiable functionals $G_1$ and $G_2$ generate the same
hamiltonian vector field, then they differ only by a constant.
\end{prop}
\begin{proof}
We have
\begin{equation}
\frac{\delta g_1}{\delta z^A} = \frac{\delta g_2}{\delta z^A}.
\end{equation}
Because $G_1$ and $G_2$ are both differentiable, it imposes
$\delta (G_1-G_2)=0$. Due to the path-connectedness of the
phase-space, the functional $G_1-G_2$ is a constant.
\end{proof}
\noindent This property relies heavily on the notion of differentiable
functional. If we drop the differentiability condition and use $Q^A =
\sigma^{AB} \frac{\delta g}{\delta z^B}$ as the definition of the
evolutionary vector field associated to $G$ then a generator would be
defined only up to a boundary term.

The following definition and properties describe the Poisson bracket
induced on the differentiable functionals by the symplectic structure
$\Omega$.
\begin{definition}
The bracket of two differentiable functionals $F$ and $G$ is the
functional given by
\begin{equation}
\left\{ F, G\right\} = \iota_F \iota_G \Omega = \Omega(G^A, F^A)
\end{equation}
where $F^A$ and $G^A$ are the characteristics of the hamiltonian
vector fields associated to $F$ and $G$:
\begin{equation}
F^A = \sigma^{AB}\frac{\delta f}{\delta z^B}, \quad G^A = \sigma^{AB}\frac{\delta g}{\delta z^B}.
\end{equation}
\end{definition}
\noindent The bracket takes the simple form:
\begin{equation}
\left\{ F, G\right\} = \int_\Sigma \frac{\delta f }{\delta z^A}
\sigma^{AB} \frac{\delta g}{\delta z^B} d^nx.
\end{equation}
\begin{prop}
The variation of a differentiable functional $F$ along the hamiltonian
vector field generated by $G$ is given by
\begin{equation}
\delta_G F = \left\{ F, G \right\}.
\end{equation}
\end{prop}
\begin{proof}
We have
\begin{eqnarray}
\delta_G F &=& \int_\Sigma \frac{\delta f}{\delta z^A} G^A d^nx +
\oint_{\d \Sigma} I^n_{G}(f d^nx) \nonumber \\
&=& \int_\Sigma \frac{\delta f}{\delta z^A} G^A d^nx = \int_\Sigma \frac{\delta f }{\delta z^A}
\sigma^{AB} \frac{\delta g}{\delta z^B} d^nx.
\end{eqnarray}
The boundary term is zero because the vector field $G^A \frac{\d}{\d z^A}$ preserves the
boundary conditions and $F$ satisfies \eqref{eq:propdiffII}.
\end{proof}
\begin{prop}
The bracket $\left\{ F, G\right\}$ defines a Poisson bracket on the
set of differentiable functionals:
\begin{itemize}
\item $\left\{ F, G\right\}$ is a differentiable functional
\item $\left\{ F, G\right\} = -\left\{ G, F\right\}$
\item $\left\{ \left\{ F, G\right\}, H\right\} + \left\{ \left\{ G,
      H\right\}, F\right\} + \left\{ \left\{ H, F\right\}, G\right\} = 0$
\end{itemize}
where $F, G$ and $H$ are differentiable functionals.
\end{prop}
\begin{proof}
We will only prove the first condition as the other two can be proved
easily following the discrete case. This proof has first been
done in \cite{Brown1986a}. We have to prove that $\delta\left\{ F,
  G\right\}$ does not contain boundary terms and that its associated
vector field preserves the boundary conditions.

One can prove that
\begin{equation}
\delta \frac{\delta f}{\delta z^A} G^A d^nx = \delta z^A \delta_G
\frac{\delta f}{\delta z^A} d^n x + d \left( \delta I^n_G (f d^n x) -
  \delta_G I^n (f d^nx)\right).
\end{equation}
Integrating over $\Sigma$, we obtain
\begin{equation}
\int_\Sigma \delta \frac{\delta f}{\delta z^A} G^A d^nx = \int_\Sigma \delta z^A \delta_G
\frac{\delta f}{\delta z^A} d^n x.
\end{equation}
The boundary terms disappear because $F$ is a differentiable
functional and $G^A \frac{\d}{\d z^A}$ preserves the boundary conditions (see
appendix \ref{sec:boundary-conditions}).
The differential of $\left\{ F, G\right\}$ is then easily computed
\begin{eqnarray}
\delta \left\{ F, G\right\} & = & \int_\Sigma \left( \delta \frac{\delta
    f}{\delta z^A} G^A - \delta \frac{\delta g}{\delta z^A} F^A\right)
d^nx \nonumber \\
&=& \int_\Sigma \delta z^A \left(\delta_G
\frac{\delta f}{\delta z^A} -\delta_F
\frac{\delta g}{\delta z^A} \right)d^n x,
\end{eqnarray}
which does not contain any boundary term. The characteristics of the
hamiltonian vector field generated by $\left\{ F, G\right\}$ can be
read off from the above equation:
\begin{eqnarray}
\left\{ F, G\right\}^A &= & \sigma^{AB}\left(\delta_G
\frac{\delta f}{\delta z^B} -\delta_F
\frac{\delta g}{\delta z^B} \right) \nonumber \\
&=& \delta_G F^A - \delta_F G^A \nonumber \\
&=& -\left[F, G \right]^A,
\end{eqnarray}
where the last expression is the Lie bracket of the two hamiltonian
vector fields $F^A \frac{\d}{\d z^A}$ and $G^A \frac{\d}{\d
  z^A}$. The Lie bracket of two evolutionary vector fields preserving
the boundary conditions preserves the boundary conditions which implies
that $\left\{ F, G \right\}$ is a differentiable functional.
\end{proof}

\begin{corollary}
The application sending a differentiable generator onto its associated
hamiltonian vector field is a homomorphism of Lie algebras.
\end{corollary}

It is important to keep in mind that only differentiable
functionals can enter the Poisson bracket. For all purposes, functionals
that are not differentiable don't exist in the hamiltonian framework. This fact has a lot of
consequences. It is, for instance, the property used to solve the problem of
charges in gauge theories. We will also use it in section \ref{sec:reduced-phase-space} to build
functionals in order to probe the reduced phase-space of theories with
no local degrees of freedom.

\subsection{Well defined Actions}
\label{sec:well-defined-action}

We saw in the previous section that the notion of differentiable
generator is a key point of the canonical formalism for field theories
defined on a manifold with boundary. This condition can be
reinterpreted as follows: 
if $G$ is a differential functional then the hamiltonian action
generating the evolution along the associated hamiltonian vector field
is well defined
\begin{equation}
S_G[z^A] = \int ds \left(\int_\Sigma d^nx \frac{1}{2} \sigma_{AB}z^A
  \d_s z^B - G[z^A]\right).
\end{equation}
We use ``well defined'' in the sense that the variation of the action
$S_G$ will not generate any boundary term on $\d \Sigma$. The
condition that $G^A \frac{\d}{\d z^A}$ preserves the boundary
conditions is equivalent to the requirement that the evolution along
the parameter $s$ stays inside the allowed configurations. Those are of
course important properties in the case where the differentiable
functional is the Hamiltonian of the theory: $H[z^A]$.

\vspace{5mm}

The canonical structure and Hamiltonian of a theory are usually deduced
from the Lagrangian description of the theory. One expects that a well
defined Lagrangian action will lead to a differentiable
Hamiltonian. This is indeed the case.

Let's assume that we have a set of boundary conditions for the
dynamical fields $\phi^a$ and a well-defined Lagrangian $\cL$ on $\cM = \RR
\times \Sigma$. As we saw in the introduction, 
the differentiation of a well-defined Lagrangian does not create any boundary term:
\begin{equation}
\delta S[\phi] = \int \int_\Sigma \left( \frac{\delta \cL}{\delta
    \phi^a} \delta \phi^a + \frac{\delta \cL}{\delta
    \dot \phi^a} \delta \dot\phi^a\right)d^nx\, dt.
\end{equation}
We have also assumed that $\cL$ does not depend on second or higher
time derivatives. The Euler-Lagrange derivatives are only defined on
$\Sigma$, they don't take into account the derivatives with respect to
$t$. We will restrict our analysis to boundary conditions on $\d \Sigma$ that are
independent of time. If this is not the case, the
phase-space is time-dependent and the canonical structure developed
in the previous section needs to be improved.

The momenta are defined as
\begin{equation}
\label{eq:deffmomenta}
\pi_a \equiv \frac{\delta \cL}{\delta\dot \phi^a}.
\end{equation}
If this relation can be inverted, we
can express $\dot \phi^a$ as local functions of $\pi_a$ and
$\phi^a$. The boundary conditions on $\phi^a$ imply boundary
conditions on $\pi_a$. The Hamiltonian is then defined as
\begin{equation}
\label{eq:defHamiltonian}
H[\phi^a, \pi_a] = \left.\int_\Sigma  \left( \pi_a \dot \phi^a -
  \cL\right)d^nx\,\right\vert_{\dot \phi^a = \dot \phi^a(\pi,\phi)}.
\end{equation}
The variation of $H$ can be easily computed:
\begin{eqnarray}
\delta H &=& \left.\int_\Sigma \left( \delta \pi_a \dot \phi^a + \pi_a \delta \dot
\phi^a - \delta \cL\right) d^nx \,\right\vert_{\dot \phi^a = \dot \phi^a(\pi,\phi)}\nonumber \\
 & =  & \left.\int_\Sigma \left( \delta \pi_a \dot \phi^a + \pi_a \delta \dot
\phi^a - \frac{\delta \cL}{\delta
    \phi^a} \delta \phi^a - \frac{\delta \cL}{\delta
    \dot \phi^a} \delta \dot\phi^a\right) d^nx \,\right\vert_{\dot \phi^a = \dot
  \phi^a(\pi,\phi)}\nonumber \\
& =  & \left.\int_\Sigma \left( \delta \pi_a \dot \phi^a - \frac{\delta \cL}{\delta
    \phi^a} \delta \phi^a\right) d^nx \,\right\vert_{\dot \phi^a = \dot
  \phi^a(\pi,\phi)}
\label{eq:diffhamil}
\end{eqnarray}
which does not contain any boundary term. The hamiltonian vector field
associated to $H$ is the time evolution. A consistent choice of
boundary conditions for $\cL$ requires these boundary conditions to be preserved by the time
evolution. The Hamiltonian $H$ is a differentiable generator by construction.

 If \eqref{eq:deffmomenta} is not invertible, we have to add primary
 constraints $\psi_\alpha = 0$ (see \cite{Henneaux1992}
for the details). The remarkable property of $H$ is
that it depends only on $\phi^a$ and $\pi_a$ even when the relation
\eqref{eq:deffmomenta} is not invertible. The analysis above is still
valid and $H[\pi,\phi]$ is again a differentiable functional.
The constraints $\psi_\alpha$ as such are not differentiable
functionals and can not enter the Poisson bracket. The solution is to
build the smeared quantities:
\begin{equation}
\label{eq:defgammaprimary}
\Gamma_\lambda[\phi^a, \pi_a] = \int_\Sigma\lambda^\alpha \psi_{\alpha} \,d^nx ,
\end{equation}
where the different possible functions $\lambda^\alpha(x)$ play the role of
labels. The differentiability of $\Gamma_\lambda$ will impose boundary conditions on
$\lambda^\alpha(x)$. The local constraints $\psi_\alpha=0$ are equivalent to the
requirement that $\Gamma_\lambda =0$ for any allowed function $\lambda^\alpha$. The
hamiltonian action is then given by: 
\begin{equation}
\label{eq:wdactiongauge}
S[\phi^a, \pi_a, \lambda^\alpha] = \int \left( \int_\Sigma \pi_a
  \dot \phi^a d^nx - H[\phi^a, \pi_a] - \Gamma_\lambda[\phi^a, \pi_a]
  \right)dt.
\end{equation}
The functions $\lambda^\alpha$ are the Lagrange multipliers
enforcing the constraints $\psi_\alpha=0$. Both $H$ and $\Gamma_\lambda$ are
differentiable generators as expected. 

The action
\eqref{eq:wdactiongauge} is not the end of the story and one can continue
the Dirac algorithm to build the full set of constraints of the
theory. The secondary constraints are obtained by requiring the
preservation in time of the primary constraints: $\frac{d}{dt}\Gamma_\rho=0$ for
all allowed $\rho$. This gives:
\begin{equation}
0 = \frac{d}{dt}\Gamma_\rho = \left\{\int_\Sigma \rho^\alpha \psi_{\alpha}\,d^nx , H[\phi^a, \pi_a] + \Gamma_\lambda[\phi^a, \pi_a]\right\}.
\end{equation}
This gives conditions on the Lagrange multipliers $\lambda^\alpha$ or
new constraints. If there are new constraints, this procedure has to
be continued to check their preservation in time by constructing the
associated differentiable smeared quantities.

\begin{example}
\label{ex:EM}
Electromagnetism in 4D.

We will work with $\Sigma$ a ball of finite radius in $\RR^3$. The metric on $\RR \times\Sigma$ is the flat
metric $\eta_{\mu\nu}$. The action is then given by:
\begin{equation}
S[A_\mu] = \int_{\RR\times\Sigma}  \frac{-1}{4}F_{\mu\nu}F^{\mu\nu}\,d^4x.
\end{equation}
It is well defined if $A_0$, $A_I$ are fixed on the
boundary where $x^I=(\theta, \phi)$ are coordinates on the sphere $\d \Sigma$ and
$(A_0, A_I)$ are the components of the pull-back of $A_\mu$ on the
boundary $\RR \times \d \Sigma$. We will take these boundary values
to be time independent. The momenta are given by:
\begin{equation}
\pi^0 = 0, \quad \pi^i = F^{i0}, \quad \pi_i = \d_0 A_i - \d_i A_0.
\end{equation}
We have one primary constraint $\psi = \pi^0$. The smeared quantity
\eqref{eq:defgammaprimary} becomes $\Gamma_\lambda = \int_\Sigma
 \lambda \pi^0\,d^3x$. Its associated hamiltonian vector field will
preserve the boundary conditions ($A_0$ fixed) only if $\lambda=0$ on
the boundary. The hamiltonian action is then
\begin{gather}
S[\pi^\mu, A_\mu, \lambda] = \int \int_\Sigma  \left(\pi^\mu
  \dot A_\mu - h -\lambda \pi^0 \right)d^3xdt,\\
h = \frac{1}{2}\pi^i\pi_i + \frac{1}{4}F_{ij}F^{ij} + \pi^i \d_i A_0.
\end{gather}
The boundary conditions are $\lambda=0$ and $A_0$, $A_I$ fixed on the
boundary. One can easily check that the Hamiltonian is a
differentiable generator.

The preservation of $\psi$ will imply a secondary constraint:
\begin{eqnarray}
0 &=& \left\{ \int_\Sigma \rho \pi^0 \,d^3x, H +
  \Gamma_\lambda\right\}\\
 &  = & \int_\Sigma \rho \, \d_i \pi^i \,d^3x.
\end{eqnarray}
This must be zero for all $\rho$ (with $\rho=0$ on the boundary) which
gives us the secondary constraint $\d_i\pi^i=0$. In this case, the Dirac algorithm
stops here and the total action can be written
\begin{gather}
S[\pi^\mu, A_\mu, \lambda_1,\lambda_2] = \int \int_\Sigma  \left(\pi^\mu
  \dot A_\mu - h -\lambda_1 \pi^0 - \lambda_2 \d_i \pi^i \right)d^3xdt,
\end{gather}
with $\lambda_1=0=\lambda_2$ on the boundary. This is not yet the
usual hamiltonian action. It is obtained by solving the constraint
$\pi^0=0$ and introducing the electric field $E^i=-\pi^i$:
\begin{gather}
S[E^i, A_\mu, \lambda_2] = \int \int_\Sigma  \left[-E^i
  \dot A_i - \left(\frac{1}{2}E^iE_i + \frac{1}{4}F_{ij}F^{ij}
  \right) + E^i \d_i \left(A_0 - \lambda_2\right) \right]d^3xdt.
\end{gather}
We can absorb $\lambda_2$ in $A_0$ which then takes the role of
lagrange multiplier for the Gauss constraint $\d_i E^i=0$. However,
only the bulk part of $A_0$ is a lagrange multiplier, its boundary
value is still fixed and contributes to the 
hamiltonian through the boundary
term:
\begin{equation}
H[z^A] \approx \int_\Sigma \left(\frac{1}{2}E^iE_i + \frac{1}{4}F_{ij}F^{ij}
  \right) d^3x - \oint_{\d \Sigma} A_0 E^r\, d\Omega^2,
\end{equation}
where we used spherical coordinates $(r, \theta, \phi)$ on $\Sigma$ with
$d\Omega^2$ the usual measure on the 2-sphere. The sign
$\approx$ denotes equality on the surface of the constraints.

In a similar way, in gravity, the lapse and shift become
lagrange multipliers after solving the primary constraints, but only
their bulk part are really arbitrary.
 Their boundary values contribute in a non-trivial way to the
Hamiltonian of the theory (see section \ref{sec:bound-cond-total}).
\end{example}

\subsection{Global Symmetries}
\label{sec:symmetries}

In the hamiltonian description of discrete systems, Noether's theorem relates symmetries of the Hamiltonian action and generators that
commute with the Hamiltonian of the theory. In this section, we will
study the same problem in field theories and its interaction with the
notion of differentiable generator. We will do this analysis without
constraints but it can be extended to gauge field theories (see
exercice 3.24 of \cite{Henneaux1992}).

\vspace{5mm}

Let's consider the following hamiltonian action:
\begin{equation}
S[z^A] = \int \left(\int_\Sigma \frac{1}{2} \sigma_{AB}z^A\dot z^B -
H[z^A]\right)\, d^nxdt,
\end{equation}
with a differentiable Hamiltonian. Let's assume that we have a
variation $\delta_Q$ preserving the boundary conditions, with $Q^A$ local functions of $z^A$ and
not of their time derivatives, such that
it also preserves the action:
\begin{equation}
\label{eq:conservationaction}
\delta_Q S = \int \int_\Sigma \, \frac{d}{d t} k\, d^nxdt, \qquad
\delta_Q z^A = Q^A.
\end{equation}
We only allowed conservation up to a total time derivative as we expect the
theory to behave exactly as a discrete mechanical system. We
will see later that this conditions is necessary (see theorem \ref{theo:symmetries-2}). Expanding the left
hand side, we obtain:  
\begin{equation}
\delta_Q S =\int \int_\Sigma \left( \frac{1}{2} \sigma_{AB}Q^A\dot z^B +
  \frac{1}{2} \sigma_{AB}z^A\frac{d}{dt} Q^B - \frac{\delta h}{\delta
    z^A} Q^A\right)\, d^nxdt.
\end{equation}
There is no boundary term coming from the variation $\delta_QH$ as $H$
is a differential functional. Let's introduce the functional:
\begin{equation}
G[z^A]=\int_\Sigma \left(\frac{1}{2} \sigma_{AB}z^AQ^B -
  k\right)\, d^nx, \quad g=\frac{1}{2} \sigma_{AB}z^AQ^B -
  k.
\end{equation}
We can rewrite condition \eqref{eq:conservationaction} as
\begin{equation}
\label{eq:conservI}
\frac{d}{dt} G = \int_\Sigma \left( - \sigma_{AB}Q^A\dot z^B + \frac{\delta h}{\delta
    z^A} Q^A\right)\, d^nx.
\end{equation}
This equality is valid for any values of the fields $z^A$ and any
values of their time derivative $\dot z^A$. We can expand the left
hand side as
\begin{eqnarray}
\frac{d}{dt} G & = & \int_\Sigma \left( \frac{\d g}{\d t} +
  \frac{\delta g}{\delta z^A} \dot z^A\right)\, d^nx + \oint_{\d \Sigma}
 I^n_{\dot z} (g d^nx)
\end{eqnarray}
where $\frac{\d }{\d t} $ is the partial derivative with respect to
$t$ and $I^n_{\dot z}$ is is the homotopy operator \eqref{eq:homotopQ} associated to
$\delta_{\dot z} z^A = \dot z^A$.
Because $\dot z^A$ is an arbitrary variation of $z^A$ and there is no
boundary term involving it in the right hand side of 
\eqref{eq:conservI}, it implies that the boundary
term in the above 
expression is zero. We recognize part of the differentiability of
the functional $G$:
\begin{equation}
\delta G = \int_\Sigma \frac{\delta g}{\delta z^A} \delta z^A\, d^nx,
\quad \Leftrightarrow \quad \oint_{\d\Sigma} I^n_{\delta z}(g d^nx).
\end{equation}
Locally, equation
\eqref{eq:conservI} gives also
\begin{eqnarray}
\frac{\delta g}{\delta z^A} & = & \sigma_{AB} Q^B,\\
\frac{\d g}{\d t} & = & \frac{\delta h}{\delta z^A} Q^A.
\end{eqnarray}
The first equation implies that $Q^A\frac{\d}{\d z^A}$ is a
hamiltonian vector field with generator $G$ and that $G$ is a differentiable functional as
$\delta_Q$ preserves the boundary conditions. The second equation gives
the conservation of $G$:
\begin{equation}
\frac{d}{dt} G = \frac{\d G}{\d t} + \left\{ G, H \right\} = 0.
\end{equation}

\begin{theorem}
If we have a variation $\delta_Q$ preserving the boundary conditions
and preserving the action in the sense \eqref{eq:conservationaction}
then there exists a differentiable generator $G$ such that
$Q^A\frac{\d}{\d z^A}$ is the associated hamiltonian vector field and satisfying
\begin{equation}
\label{eq:conservationG}
\frac{\d G}{\d t} + \left\{ G, H \right\} = 0.
\end{equation}
\end{theorem}

\vspace{5mm}

The converse is also true. Let's assume that we have a conserved differentiable
functional $G$, i.e. satisfying equation
\eqref{eq:conservationG}. Then the variation of the action along the
associated hamiltonian vector field $G^A\frac{\d}{\d z^A}$ is given by:
\begin{eqnarray}
\delta_G S &=& \int \left[\int \left(\frac{1}{2} \sigma_{AB}G^A\dot
  z^B + \frac{1}{2} \sigma_{AB}z^A\frac{d}{dt}G^B\right) d^nx - \delta_G
H\right]\,dt\nonumber \\
&=& \int \left[\int \left(\sigma_{AB}G^A\dot
  z^B + \frac{1}{2} \sigma_{AB}\frac{d}{dt}(z^AG^B)\right) d^nx - \left\{H,
G\right\}\right]\,dt\nonumber \\
&=& \int \left[\int \left(-\frac{\delta g}{\delta z^A}\dot
  z^A  - \frac{\d g}{\d t}\right)d^nx  + \frac{d}{dt}\int \frac{1}{2}
  \sigma_{AB} z^AG^B\,d^nx\right]\,dt\nonumber \\
&=& \int \frac{d}{dt} \left[ G+ \int \frac{1}{2} \sigma_{AB} z^AG^B\,d^nx \right]\,dt,
\end{eqnarray}
where we have used the fact that $\dot z^A$ is a variation of the
fields $z^A$ preserving the boundary conditions to obtain the last
line. 

\begin{theorem}
\label{theo:symmetries-2}
If $G$ is a differentiable functional such that
\begin{equation}
\label{eq:conservationGII}
\frac{\d G}{\d t} + \left\{ G, H \right\} = 0,
\end{equation}
then the variation generated by $G$ preserves the action in the sense
of \eqref{eq:conservationaction}.
\end{theorem}

\vspace{5mm}

These results are a direct application of Noether's theorem but they
are free of the problems we presented in the introduction. The
conserved charges that we built are defined up to a constant only
and not up to a boundary term. Applying Noether's theorem to a gauge
symmetry will construct the associated generator with the right
boundary term. The source of the problem of gauge theories has been
handled by a careful treatment of the boundary
conditions. We will see in the next section how to apply these ideas to
the computation of surface charges.

%%%%%%%%%%%%%%%%%%%%%%%%%%%%%%%%%%%%%%%%%%%%%%%%%%%%%%%%%%%%%%%%%%%%%%

\newpage
\section{Applications to Gauge Theories}
\label{sec:surface-charges}

As we saw in the introduction, surface charges are  conserved quantities
associated to gauge-like transformations. For instance, in electromagnetism, we
have the electric charge or, in gravity, the
energy and the angular momentum of the system. In both cases, these
conserved quantities are associated with transformations that look
like gauge transformations. As they are generated by constraints, we expect
these charges to give zero on the constraints surface. We
will see in this section how the notion of differentiable generator introduced in the previous
section solves the problem and associates non-zero charges to a certain
class of gauge-like transformations. 

\vspace{5mm}

Usually
when we want to compute the surface charges of a theory, we have a
set of solutions for which we want to compute those charges. This
means that we only know
the local form of the action of the theory and not all the boundary terms
necessary to make it well defined. In that spirit, we will consider theories of the form
\begin{equation}
\label{eq:SCaction}
S[z^A, \lambda^\alpha] = \int \int \left( \frac{1}{2}
  \sigma_{AB}z^A\dot z^B - h - \lambda^\alpha \gamma_{\alpha} \right)\, d^nxdt,
\end{equation}
where $h$ is the first class hamiltonian density and $\gamma_{\alpha}$
are the full set of
first-class constraints. The weak equality sign $\approx$ will be used
for the equality on the constraints surface $\gamma_\alpha \approx 0$. We will denote $h_T = h + \lambda^\alpha
\gamma_{\alpha}$ the total hamiltonian density. Such an action is in general well defined only for very
restrictive boundary conditions that usually don't contain the
solutions we are interested in.

The first step in the analysis is to
define a set of boundary conditions containing the solutions of
interest and add the right boundary term to \eqref{eq:SCaction} in
order to make everything well defined.

\subsection{Boundary Conditions and Total Hamiltonian}
\label{sec:bound-cond-total}

The choice of boundary conditions is a very tricky one and there are a
lot of different possibilities. The only 
restriction imposed by the consistency of the theory is that the
total Hamiltonian be a differential functional. When dealing with boundaries
at infinity, we will also require finiteness of the total Hamiltonian.

If we want to compute the charges of some particular
solutions associated to specific 
symmetries, we need boundary conditions that both contain the solutions and
are preserved by the symmetries under consideration. 

\vspace{5mm}

Let's assume that a set of boundary conditions for both the dynamical fields $z^A$
and for the lagrange multipliers $\lambda^\alpha$ has been
selected. As in section \ref{sec:well-defined-action}, we need the
boundary conditions on the dynamical fields $z^A$ to be time
independent. The differentiability of $H_T$ implies two conditions. The first one
is that its associated evolutionary vector field $\sigma^{AB} \frac{\delta 
  h_T}{\delta z^B}$ preserves the boundary conditions on $z^A$. The
second condition is that there exist a $(n-1,0)$-form $b(z, \lambda)$
such that the total Hamiltonian defined by
\begin{equation}
H_T = \int_\Sigma h_T d^nx + \oint_{\d \Sigma} b,
\end{equation}
is a differentiable functional. The boundary term must satisfy 
\begin{equation}
- \oint_{\d \Sigma} \delta b = \oint_{\d \Sigma} I^n (h_T d^nx),
\end{equation}
where the right hand side is the boundary term produced by the
variation of the bulk term $\int_\Sigma h_T d^nx$. As we saw in the previous section (proposition
\ref{pr:uptoaconst}), this boundary term is 
defined only up to a constant with respect to $z^A$: it is defined up
to a function of the lagrange multipliers $\lambda^\alpha$.

\vspace{5mm}

The candidate action is then given by
\begin{eqnarray}
\label{eq:SCactionFull}
S[z^A, \lambda^\alpha] &=& \int \left[\int \left( \frac{1}{2}
  \sigma_{AB}z^A\dot z^B - h - \lambda^\alpha \gamma_{\alpha} \right) d^nx
- \oint_{\d \Sigma} b\right]dt \nonumber \\ &=& \int \left( \int \frac{1}{2}
  \sigma_{AB}z^A\dot z^B \,d^nx -  H_T[z^A, \lambda^\alpha]\right)dt.
\end{eqnarray}
The above considerations make this action well defined with respect to
variations of the dynamical fields $z^A$. With this, the hamiltonian
structure of the theory is well defined and we can continue the
analysis. 

In general, a variation with respect to the lagrange multipliers
$\lambda^\alpha$ will still produce boundary term in the action which
lead to extra constraints on the boundary. It will be particularly
useful in section
\ref{sec:BC_on_the_phase-space} to encode part of
the boundary conditions as extra constraints. 
However, in most cases, those boundary
constraints are not welcome. If one wants to remove them, one needs to select the
total Hamiltonian so that the action is also well 
defined with respect to variations of the lagrange multipliers:
\begin{equation}
\delta^\lambda S[z^A, \lambda^\alpha] =\int \left[\int
  \, \delta^\lambda\lambda^\alpha \gamma_{\alpha} \, d^nx 
+ \oint_{\d \Sigma}\delta^\lambda b\right]dt, 
\end{equation}
where $\delta^\lambda$ only hits the lagrange multipliers. It means
that $b$ must be independent of $\lambda^\alpha$. If such a $b$  exists, this
last requirement fixes the boundary term up to a constant. In the following of
this section, we will assume that such a boundary term has been added and that
there are no extra boundary constraints.

\vspace{5mm}

For certain sets of boundary conditions, it might not be possible
to find a boundary term satisfying all the above. This is
referred to as the integrability
problem. In that case, the selected boundary conditions are too relaxed and it
is not possible to write an associated well defined theory. The only
solution is to restrict the boundary conditions.

In the following, we will assume that a set of boundary conditions for
both $z^A$ and $\lambda^\alpha$ has
been selected such that the total Hamiltonian
$H_T$ is a differentiable generator. In general, the boundary
conditions of $\lambda^\alpha$ can depend on $z^A$. It will be useful to
decompose the lagrange multipliers $\lambda^\alpha$ as $\lambda^\alpha = \bar \lambda^\alpha + \mu^\alpha$
such that $\bar \lambda^\alpha$ are fixed in term of $z^A$
 in order to encode this dependence. The quantities
$\mu^\alpha$ are left to vary freely up to boundary conditions
independent of $z^A$ with:
\begin{equation}
\quad H_T[z^A, \lambda^\alpha] -
H_T[z^A, \bar\lambda^\alpha(z^A)] =
\int_\Sigma \mu^\alpha\gamma_\alpha \,d^nx \approx 0.
\end{equation}
The fields $\mu^\alpha$ are the real lagrange multipliers and
$\bar\lambda^\alpha(z^A)$ encode the contribution of $\lambda^\alpha$ to
$H_T$ through the boundary term.

We will also assume that among the boundary conditions we impose the
constraints and all their derivatives on the boundary. This does not remove
any extra degrees of freedom from the theory and it simplifies some
computations.

\vspace{5mm}

\begin{example}
Gravity.

The local action for gravity in $n+1$ dimensions is given by \cite{Arnowitt2008,Regge1974}:
\begin{eqnarray}
S[\pi^{ij}, g_{ij}, N, N^i] &=& \int \int_{\Sigma} \left( \pi^{ij}
  \dot g_{ij} - N \cR - N^i \cR_i\right) d^nxdt,\\
\cR & = & -\sqrt{g} R - \frac{1}{\sqrt{g}} \left(\frac{1}{n-1} \pi^2 -
  \pi_{ij}\pi^{ij}\right),\\
\cR_i & = & -2 \nabla_j \pi^j_i,
\end{eqnarray}
where indices are lowered and raised using the metric $g_{ij}$ and its
inverse $g^{ij}$. The derivative $\nabla_i$ is the covariant derivative associated to $g_{ij}$
and $R$ is the corresponding Ricci scalar. The tensor $\pi^{ij}$ is
treated as a density. The lapse $N$ and shift $N^i$ are coming from
the $3+1$ decomposition of the 4 dimensional metric. We
will assume that $\Sigma$ is a finite manifold with a boundary
$\partial \Sigma$. The
boundary term adapted to the Dirichlet boundary conditions in the
hamiltonian formalism can be found in
\cite{Brown1993}. The boundary conditions are given
by: 
\begin{gather}
N^r\vert_{\partial \Sigma} = 0, \qquad N\vert_{\partial \Sigma} = \bar
N, \\ N^A\vert_{\partial \Sigma} = \bar N^A, \qquad
g_{AB}\vert_{\partial \Sigma} = \gamma_{AB},
\end{gather}
where the coordinates are given by $x^i=r, x^A$ and the boundary
$\partial\Sigma$ is a surface at constant $r$. Fixing $N, N^A, g_{AB}$
is equivalent to fixing $g_{00}, g_{0A}, g_{AB}$ when $N^r=0$. We will assume
these quantities to be time independent.
The surface term needed to make the total Hamiltonian well defined is then
easily computed. The variation of the smeared constraints gives
\begin{multline}
-\oint_{\d\Sigma} I^n \left( N \cR + N^i \cR_i\right) = \\ \oint_{\partial \Sigma} 
 \left\{2 N^A \delta \pi^r_{\phantom k A}  +
  \sqrt{\gamma} N \left( \delta K + \gamma^{AB} \delta K_{AB} \right)\right\} d^{n-1}x,
\end{multline}
where $K$ is the trace of the extrinsic curvature of the boundary.
The total Hamiltonian is then defined as:
\begin{equation}
H_T = \int_{\Sigma}  \left( N \cR + N^i \cR_i\right)d^nx + 2 \oint_{\partial \Sigma} 
\left( N^A \pi^r_{\phantom k A}  +
  N \sqrt{\gamma} K\right) d^{n-1}x.
\end{equation}
When
evaluated on a solution, the total Hamiltonian gives the energy of the
system. In this case, the only non-zero contribution comes from the
boundary term and we have 
\begin{equation}
H_T \approx  2 \oint_{\partial \Sigma} 
 \left(\bar N^A \pi^r_{\phantom k A}  +
 \bar N  \sqrt{\gamma}  K\right) d^{n-1}x.
\end{equation}
This value of the Hamiltonian is tied to the boundary conditions we
imposed on $N$ and $N^i$. This reflects the fact that those boundary
values are not behaving as lagrange multipliers but carry information
about the dynamics of the system.
\end{example}

\subsection{Differentiable Gauge Transformations}
\label{sec:Diff-gauge}

Gauge transformations are transformations generated by
the first-class constraints of the theory through the Poisson
bracket. In field theories, we call gauge transformation any
transformation of the form
\begin{equation}
\label{eq:gaugez}
\delta_{\eta} z^A = \sigma^{AB}\frac{\delta}{\delta
z^B}( \eta^\alpha \gamma_\alpha),
\end{equation}
where the gauge parameters $\eta^\alpha$ can be functions of the
dynamical fields $z^A$. 
The algebra of these transformations closes:
\begin{equation}
\label{eq:gaugealg}
\left[\delta_\eta, \delta_\rho\right] = \delta_\epsilon, \qquad \left[
  \eta, \rho\right]^\alpha_g \equiv \epsilon^\alpha
\end{equation}
where $\delta_\epsilon$ is a gauge transformation and $\left[
  \eta, \rho\right]^\alpha_g$ is the bracket induced on the gauge
parameters. 
 The transformations \eqref{eq:gaugez} can be extended to the
lagrange multipliers $\lambda^\alpha$ in order to leave the action
invariant up to a boundary term \cite{Henneaux1992}. The resulting variation is:
\begin{equation}
\label{eq:gaugel}
\delta_\eta \lambda^\alpha = \frac{\d}{\d t} \eta^\alpha + [\lambda,
\eta]^\alpha_g - V(\eta)^\alpha.
\end{equation}
The quantity $V(\eta)^\alpha$ is defined from the first class
hamiltonian $h$ by 
\begin{equation}
\delta_\eta h = V(\eta)^\alpha\gamma_{\alpha},
\end{equation}
where the equality is up to boundary terms.
In general, gauge transformations describe the
redundancy of the description but, as we saw in the introduction, in the
presence of a spatial boundary the story is different. In order to
avoid confusion, we will call the transformations
\eqref{eq:gaugez}-\eqref{eq:gaugel} gauge-like transformations.

\vspace{5mm}

The above considerations are only local and don't take into account
the boundary structure of the theory. As we saw in section
\ref{sec:Symplectic-struct}, only differentiable functionals are
allowed in the Poisson bracket, however there is no guarantee that
gauge-like transformations are generated by differentiable functionals.

For this to happen, we saw in section \ref{sec:Symplectic-struct}
that one needs two requirements to be satisfied. The first is the preservation of the boundary
conditions of $z^A$:
they must transform allowed configurations into allowed
configurations of the fields. Requiring that the transformations
\eqref{eq:gaugez} preserve the boundary conditions
of $z^A$ will impose boundary conditions on the
gauge parameters $\eta^\alpha$.  Let's remark that imposing the constraints
and their derivatives on the boundary don't restrict the set of allowed
gauge-like transformations.

An important observation here is that
we don't need to require the preservation of the boundary conditions
of the lagrange multipliers in order to build differentiable generators. We will see in section
\ref{sec:asympt-symm} that this additional restriction selects
 transformations that are also symmetries of the theory. This confirms
 the fact that the boundary conditions of $\lambda^\alpha$ contains
 dynamical information.

The second condition comes from requiring the existence of a suitable
boundary term to complete the generator. Its bulk part is given
by:
\begin{equation}
\bar\Gamma_\eta[z^A] = \int_\Sigma \eta^\alpha \gamma_\alpha\, d^nx.
\end{equation}
We need a $(n-1,0)$-form $k_\eta$ such that:
\begin{equation}
\oint_{\d \Sigma} \delta k_\eta = - \oint_{\d \Sigma} I^n (\eta^\alpha \gamma_\alpha d^nx).
\end{equation}
If we can find such form $k_\eta$, then the generator defined by
\begin{equation}
\label{eq:gaugegenefull}
\Gamma_\eta[z^A]  = \bar \Gamma_\eta[z^A]  + \oint_{\d \Sigma} k_\eta
\end{equation}
is differentiable.
It is not always possible to find such a boundary term and the 
integrability problem might also appear here. In that case, only a subset of the allowed
gauge-like transformations will have an associate differentiable
generator. We will call this subset of gauge-like transformations
differentiable gauge transformations.

\vspace{5mm}

As the two requirements we added for differentiable gauge
transformations are preserved by the bracket of evolutionary vector fields,
the differentiable gauge transformations form a subalgebra of the
algebra of gauge-like transformations \eqref{eq:gaugealg}. The
algebra of the associated generators forms a representation of this subalgebra:
\begin{theorem}
\label{theo:algebraI}
Let the phase-space $\cF(\Sigma)$ be path-connected. 
If $\eta^\alpha$ and $\rho^\alpha$ are two differentiable gauge
transformations then:
\begin{equation}
\label{eq:algecharges}
\left\{ \Gamma_\eta[z^A], \Gamma_\rho[z^A]\right\} =
\Gamma_{[\rho,\eta]_g}[z^A] + \cK_{\eta, \rho}
\end{equation}
where $\cK_{\eta, \rho}$ is a possible central extension. It satisfies
\begin{gather}
\label{eq:centralext}
\cK_{\eta, \rho} = -\cK_{\rho, \eta},\\
\cK_{[\eta, \rho]_g, \epsilon}+\cK_{[\rho, \epsilon]_g, \eta}+\cK_{[\epsilon, \eta]_g, \rho}=0,
\end{gather}
for all differentiable gauge transformations $\eta^\alpha$,
$\rho^\alpha$ and $\epsilon^\alpha$.
\end{theorem}
\begin{proof}
The proof of this theorem is direct. We know that the algebra of the
hamiltonian transformations is homomorphic to the algebra of the
Hamiltonian generators. The commutator of two differentiable gauge
transformations is a gauge-like transformation $ [\eta, \rho]^\alpha_g$, the Poisson bracket of the
associated generators will be a differentiable generator associated to
its resulting gauge-like transformation $ [\eta, \rho]^\alpha_g$.
The possibility of a central extension comes
from the fact that hamiltonian generators are defined up to a
constant \cite{Arnold1989}.
\end{proof}
\noindent If one drops the hypothesis of path-connectedness of the
phase-space, the extension of the algebra can become field dependent:
it could take different values on different path-connected components of the phase-space.

\vspace{5mm}

Differentiable gauge transformations split into two categories, the proper and
improper gauge transformations:
\begin{itemize}
\item proper gauge transformations $\delta_\eta$ are defined by
  $\Gamma_\eta[z^A] \approx 0$ for all configurations on the
  constraints surface. These transformations are generated by
  constraints of the theory.
\item improper gauge symmetries $\delta_\eta$ are those for which there exist
  field configurations on the constraints surface such that
  $\Gamma_\eta[z^A] \ne 0$.
\end{itemize}
There is a common misconception that improper gauge transformations
are still generated by constraints, first or second class. This is not
the case since, for improper gauge transformations, the integral
$\int_\Sigma \eta^\alpha \gamma_\alpha \,d^nx$ is not a
differentiable functional: it can not enter the Poisson bracket.

The proper gauge transformations form an ideal subalgebra of the
differentiable gauge transformations: if $\eta^\alpha$ is a proper
gauge transformation and $\rho^\alpha$ is a differential gauge
transformation, the commutator of the two gauge transformations
$[\eta, \rho]_g$ is
generated by
\begin{equation}
\left\{\Gamma_\rho[z^A], \Gamma_\eta[z^A] \right\} = -\delta_\rho
\Gamma_\eta[z] \approx 0.
\end{equation}
This is zero on the constraint surface because
$\Gamma_\eta[z^A]\approx 0$ and $\delta_\rho$ preserves the
constraints. 

\vspace{5mm}

Up until now, we have been very careful of not talking about
symmetries. The differentiable gauge transformations are in general
not symmetries of the theory: they don't satisfy
\begin{equation}
\label{eq:consGamma}
\frac{\d}{\d t} \Gamma_{\eta} + \left\{\Gamma_{\eta}, H_T \right\}
\approx 0.
\end{equation}
The weak equality here denotes equality on the constraints surface and
is the conservation condition \eqref{eq:conservationGII} for gauge
theories \cite{Henneaux1992}.

However, proper gauge transformations are symmetries of the system. It comes
from the fact that the evolutionary vector field associated to $H_T$
transforms allowed configurations into allowed configurations and
preserves the constraints which implies
\begin{equation}
\left\{\Gamma_{\eta}, H_T \right\} \approx 0
\end{equation}
for proper gauge transformations $\delta_\eta$. Combining it with
$\Gamma_{\eta}\approx 0 \Rightarrow\frac{\d}{\d
  t}\Gamma_{\eta}\approx 0$ we obtain the conservation condition
\eqref{eq:consGamma}. 

The proper gauge transformations are the real gauge symmetries of the
system. They represent the redundancy of the description. On the other
hand, part of the
improper gauge transformations will satisfy \eqref{eq:consGamma} and
also be symmetries of the system. Those transformations are
global symmetries of the theory, they change the state of the
system. They are called asymptotic symmetries and will be the subject
of  section \ref{sec:asympt-symm}.

\vspace{5mm}

The fact that proper gauge transformations form an ideal of the
differentiable gauge transformations leads to a very interesting
result
\begin{theorem}
The differentiable gauge transformations are first-class quantities.
\end{theorem}
First-class functionals in gauge theories are the building
blocks. They are gauge invariant and, as such, are the observables of
the theory. In most theories, the set of first-class differentiable functionals
that one can build in this way is rather small and does not bring a
lot of information, see example below. However, in theories with no
local degrees of freedom, this set is a lot bigger and we will see in a
few examples in section \ref{sec:reduced-phase-space} that, up to
topological information, it can
describe completely the reduced phase-space of the theory.

\begin{example}
\label{firstclassEM}
Electromagnetism in 4D

As we saw in example \ref{ex:EM}, the action is written:
\begin{equation}
S[E^i, A_i, A_0] = \int \int_\Sigma \left[-E^i \dot A_i -
  \left(\frac{1}{2}E^iE_i + \frac{1}{4}F_{ij}F^{ij} \right) +
  E^i \d_i A_0\right] d^3xdt,
\end{equation}
where the values of $A_0$ and $A_I$ are fixed on the boundary
$\d \Sigma$. The gauge-like transformations are given by:
\begin{equation}
\delta_\eta A_j = \frac{\delta (\eta \d_i E^i)}{\delta E^j} \approx \d_j \eta.
\end{equation}
Preservation of the boundary conditions for $A_i$ imposes that $\d_A
\eta \approx 0$ on the boundary $\d \Sigma$. The gauge parameter must
tend to a constant $\eta_R$ on the boundary when the constraints are
satisfied. Those  transformations are generated by:
\begin{equation}
\Gamma_\eta = \int_\Sigma \eta \d_i E^i \,d^3x- \oint_{\d
  \Sigma} \eta E^r \approx - \eta_R \, Q\, d\Omega^2,
\end{equation}
where $Q$ is the electric charge  $Q = \oint_{\d
  \Sigma} E^r\, d^2\Omega$.
\end{example}

\subsection{Classification of Boundary Conditions}
\label{sec:poss-bound-cond}

As we saw in section \ref{sec:CF-FieldTheories}, only the boundary
conditions for $z^A$ are part of the definition of the canonical
structure of the theory. The boundary conditions on the lagrange
multipliers $\lambda^\alpha$ appear for the definition of the total
Hamiltonian only: they contain dynamical information.

Let's assume that, for fixed boundary conditions on $z^A$, we have
two different sets of boundary conditions for the 
lagrange multipliers, $\lambda^\alpha = \bar 
\lambda^\alpha_1+\mu^\alpha$ and $\lambda^\alpha = \bar
\lambda^\alpha_2+\mu^\alpha$, leading to two differentiable total
Hamiltonians, $H_{T1}$ and $H_{T2}$. The boundary
conditions for $\mu^\alpha$ are the same as it will span the set of
proper gauge transformations which is independent of $H_T$.
The difference $H_{T1}-H_{T2}$ is a differentiable generator and its
bulk term is given by a sum of constraints: it is the generator of a differentiable
gauge transformation.

Conversely, let's assume that we have boundary conditions for
$\lambda^\alpha$, their associated hamiltonian $H_T$ and a
differentiable generator $\Gamma_\eta$. We can build new boundary
conditions for $\lambda^\alpha$ such that the new hamiltonian is given
by $H_T+\Gamma_\eta$. The answer is obviously:
\begin{equation}
\bar\lambda_2^\alpha= \bar\lambda^\alpha + \eta^\alpha, \qquad \lambda^\alpha = \bar\lambda^\alpha + 
\eta^\alpha + \mu^\alpha, 
\end{equation}
leading to
\begin{equation}
H_{T2} = H_T +\Gamma_\eta.
\end{equation}
The new hamiltonian being the sum of two differentiable generators is
differentiable. If $\Gamma_\eta$ is the generator of a proper gauge
transformation, the two hamiltonian are equivalent: $H_{T2} \approx
H_T$.

\begin{theorem}
For gauge theories of the form \eqref{eq:SCaction}, once boundary
conditions for the dynamical variables $z^A$ have been selected, the
possible boundary conditions for the lagrange multipliers
$\lambda^\alpha$ are in one to one correspondance with the
differentiable gauge transformations modulo proper gauge transformations.
\end{theorem} 

The various theories obtained that way share the same local form of the action
\eqref{eq:SCaction}. However, the value of their total Hamiltonian
will be different due to a different boundary term. The theories
obtained are different.

\subsection{Asymptotic Symmetries}
\label{sec:asympt-symm}

We saw in section \ref{sec:Diff-gauge} that proper gauge transformations
are symmetries of the theory. As they are generated by functionals
that are zero on the constraint surface they are the set of true gauge
transformations of the theory. We will now study the set of asymptotic
symmetries which are the global symmetries of the theory hidden inside
the set of improper gauge transformations.

The necessary and sufficient condition for an improper gauge
transformation $\delta_\eta$ to be a symmetry is that its generator $\Gamma_\eta[z]$
satisfies equation \eqref{eq:conservationGII}:
\begin{equation}
\label{eq:consGammaII}
\frac{\d}{\d t} \Gamma_{\eta} + \left\{\Gamma_{\eta}, H_T \right\}
\approx 0.
\end{equation} 
The only explicit time
dependance of $\Gamma_\eta[z^A]$ is in the gauge parameter
$\eta^\alpha$:
\begin{equation}
\frac{\d}{\d t} \Gamma_\eta[z^A] = \Gamma_{\dot \eta}[z^A]
\end{equation}
and, because $\Gamma_\eta[z^A]$ is differentiable, $\Gamma_{\dot
  \eta}[z^A]$ is also differentiable. It implies that 
the functional $F$ appearing in \eqref{eq:consGammaII} 
is differentiable:
\begin{eqnarray}
\label{eq:consesurfchar_bis}
F[z^A, \eta^\alpha, \lambda^\alpha]& \equiv & \Gamma_{\dot \eta}[z^A] +
\left\{\Gamma_\eta[z^A], H_T[z^A, \lambda^\alpha]
\right\}.
\end{eqnarray}
By definition of the variation of the lagrange multipliers under a
gauge transformation \eqref{eq:gaugel}, we know that the bulk term of $F$ is
given by $\int_\Sigma \delta_\eta \lambda^\alpha \gamma_\alpha
\,d^nx$: it is the
generator of the differentiable gauge transformation with parameter
$\delta_\eta \lambda$:
\begin{equation}
F[z^A, \eta^\alpha, \lambda^\alpha] = \Gamma_{\delta_\eta \lambda}[z^A].
\end{equation}
We have proved the following result:
\begin{theorem}
An improper gauge transformation $\delta_\eta$ is a symmetry of the
theory if and only if the differentiable gauge transformation generated
by $\Gamma_{\delta_\eta \lambda}[z^A]$ is a proper gauge
transformation. 
\end{theorem}

The set of differentiable gauge symmetries forms a subalgebra of the
differentiable gauge transformations and the proper gauge transformations
form an ideal of this subalgebra. This leads to the standard definition:
\begin{definition}
The algebra obtained by taking the quotient of the
algebra of the differentiable gauge symmetries by the proper
gauge symmetries is called the asymptotic symmetry algebra.
\end{definition}
\noindent As we said earlier, it is the algebra of global symmetries of the theory hidden inside the set
of gauge-like transformations. Their differentiable generators $\Gamma_\eta$ are
constants of motion given by boundary terms when evaluated on the constraints
\begin{equation}
\Gamma_\eta[z^A] \approx \oint_{\d \Sigma} k_\eta.
\end{equation}
Those quantities are called surface charges. They form a
representation of the asymptotic symmetry algebra through the Poisson bracket:
\begin{theorem}
Let the phase-space $\cF(\Sigma)$ be path-connected. 
If $\eta^\alpha$ and $\rho^\alpha$ are two differentiable gauge
symmetries then:
\begin{equation}
\label{eq:algechargesII}
\left\{ \Gamma_\eta[z^A], \Gamma_\rho[z^A]\right\} =
\Gamma_{[\rho,\eta]_g}[z^A] + \cK_{\eta, \rho}
\end{equation}
where 
$\cK_{\eta, \rho}$ is a possible central extension. It satisfies
\begin{gather}
\label{eq:centralextII}
\cK_{\eta, \rho} = -\cK_{\rho, \eta},\\
\cK_{[\eta, \rho]_g, \epsilon}+\cK_{[\rho, \epsilon]_g, \eta}+\cK_{[\epsilon, \eta]_g, \rho}=0,
\end{gather}
for all differentiable gauge symmetries $\eta^\alpha$,
$\rho^\alpha$ and $\epsilon^\alpha$.
\end{theorem}
\noindent As for theorem \ref{theo:algebraI}, if one drops the hypothesis of path-connectedness of the
phase-space, the extension of the algebra can become field dependent.

\vspace{5mm}

A subset of those symmetries can be computed easily. Let's assume that
$\delta_\eta \lambda^\alpha$ preserve the boundary conditions of
$\lambda^\alpha$. This is a new condition: it was not required for the
differentiability of the generator $\Gamma_\eta[z^A]$. This means that
$\delta_\eta \lambda^\alpha$ satisfies the same boundary conditions
as the true lagrange multipliers $\mu^\alpha$. The generator $\Gamma_eta[z^A]$
is then simply given by 
\begin{equation}
\Gamma_{\delta_\eta \lambda} = \int_\Sigma \, \delta_\eta
\lambda^\alpha\gamma_alpha \,d^nx  \approx 0,
\end{equation}
which is the symmetry condition. The possible dependence on $z^A$ of
$\delta_\eta \lambda^\alpha$ does not produce any boundary term. Its
variation is proportional to the constraints which we imposed on the
boundary as part of our boundary conditions.

\vspace{5mm}

We see that requiring the preservation of the boundary conditions of
the lagrange multipliers guarantees that
differentiable gauge transformations become symmetries. This subset
of the asymptotic symmetries is the one usually computed.
In practice, one computes the set of gauge-like
transformations preserving both the boundary conditions of $z^A$ and
$\lambda^\alpha$. This computation can be done in the lagrangian
formalism where it is easier because we now treat both dynamical field
and lagrange multipliers in the same way. One then restricts to the
subset of those transformations generated by a differentiable
generator. This technique has been applied in many different cases to
compute both the charges and the algebra under the Poisson
bracket. 

Among the most famous examples, we find the original computation of the Poincar\'e charges for asymptotically flat space
times in 4D by Regge and Teitelboim \cite{Regge1974}, where they recover the ADM definition
of the mass \cite{Arnowitt2008}, and the computation of the
asymptotic symmetry algebra by Brown and Henneaux for asymptotically
$AdS_3$ space-times \cite{Brown1986}. An interesting non-trivial
integrability problem appeared in the study of gravity coupled to scalar
fields \cite{Henneaux2004,Henneaux2007}.

%%%%%%%%%%%%%%%%%%%%%%%%%%%%%%%%%%%%%%%%%%%%%%%%%%%%%%%%%%%%%%%%%%%%%%%%%%%%%%%%%%%%%%%%

\subsection{Brown-York Quasi-local Charges}
\label{sec:BYquasilocal}

We will now, as an example, apply the technique presented in the
previous section to the computation
of the Brown-York quasi-local charges \cite{Brown1993}.

The boundary conditions have been introduced in section
\ref{sec:bound-cond-total}. They are given by:
\begin{gather}
N^r\vert_{\partial \Sigma} = 0, \qquad N\vert_{\partial \Sigma} = \bar
N, \\ N^A\vert_{\partial \Sigma} = \bar N^A, \qquad
g_{AB}\vert_{\partial \Sigma} = \gamma_{AB},
\end{gather}
where all these quantities are time-independent.
The associated differentiable Hamiltonian is
 \begin{eqnarray}
 S[g_{ij}, \pi^{ij}, N, N^i] & = & \int \left\{
     \int_\Sigma \pi^{ij} \partial_t g_{ij} \, d^nx - H_T \right\}dt, \\
 H_T & = & \int_\Sigma \left(N \cR + N^i \cR_i \right) d^nx\\ &&
 \qquad \nonumber + 2 \oint_{\partial \Sigma} 
  \left( N^A \pi^r_{\phantom k A}  +
  N  \sqrt{\gamma}  K\right) d^{n-1}x.
\end{eqnarray}
Following the result of section \ref{sec:asympt-symm}, all
differentiable gauge transformations preserving the boundary
conditions on the lagrange multipliers $N, N^i$ give rise to conserved
quantities. The gauge-like transformations
take the form:
\begin{eqnarray}
\label{eq:gaugeN}
\delta_\xi N & = & \partial_t \xi - \left[N, \xi \right]^\perp_{SD}, \\
\label{eq:gaugeNi}
\delta_\xi N^i & = & \partial_t \xi^i - \left[N, \xi \right]^i_{SD}, \\
\label{eq:gaugeg}
\delta_\xi g_{ij} & = & 2
\frac{\xi}{\sqrt{g}} \left( \pi_{ij} - \frac{\pi}{n-1}g_{ij}\right) +
\nabla_i \xi_j + \nabla_j \xi_i.\\
\label{eq:gaugepi}
\delta_\xi \pi^{ij} & = & -\sqrt{g} \xi
\left(G^{ij} + \Lambda g^{ij} \right) + \sqrt{g} \left(
  \nabla^i\nabla^j \xi - g^{ij} \nabla^k\nabla_k\xi\right)\nonumber\\
&& - \frac{\xi}{\sqrt{g}}\left(2 \pi^{ik}\pi_k^j - \frac{2}{n-1}
  \pi\pi^{ij} \right) -
\frac{\xi}{\sqrt{g}}\frac{g^{ij}}{2}\left(\frac{1}{n-1}\pi^2 -
  \pi^{kl}\pi_{kl} \right),\nonumber\\
&& +\nabla_k \left( \xi^k \pi^{ij}\right) - \nabla_k\xi^i \pi^{kj} - \nabla_k\xi^j \pi^{ki}.
\end{eqnarray}
For simplicity, we will only consider the transformations with
$\xi^r\vert_{\d \Sigma}=0$. Preservation of the boundary conditions imposes:
\begin{eqnarray}
0 & = & \partial_t \xi - \bar N^A\partial_A \xi + \xi^A\d_A \bar N, \\
0 & = & \partial_t \xi^A -\bar N^B\partial_B \xi^A + \xi^B\d_B\bar N^A
- g^{Aj}\left(\bar N\partial_j \xi - \xi\d_j N\right), \\
\label{eq:gaugegI}
0 & = & - g^{rj}\left(\bar N\partial_j \xi - \xi\d_j N\right), \\
\label{eq:gaugegII}
0 & = & 
\frac{\xi}{\bar N} \left(  - \bar N^C\d_Cg_{AB} - \d_A
  \bar N^Cg_{CB} - \d_B \bar N^Cg_{CA}\right) \nonumber \\ && \qquad
 + \xi^C\d_Cg_{AB} + \d_A
  \xi^Cg_{CB} + \d_B \xi^Cg_{CA},
\end{eqnarray}
where all the equalities hold on the boundary. In term of the induced
metric $\gamma_{\alpha \beta}$ on the space-time boundary $\partial \Sigma \times \RR$ and the vector
$\eta^\alpha=(\frac{\xi}{\bar N}, \xi^A - \frac{\bar N^A}{\bar N}\xi)$, they become
\begin{gather}
\gamma_{\alpha\beta} = \left( \begin{array}{cc} -\bar N^2 +\bar  N_c\bar  N^c &\bar  N_B
    \\ \bar N_A & \gamma_{AB}
  \end{array}\right),\\
\label{eq:liegammaBY}
\eta^\alpha \d_\alpha \gamma_{\beta\delta} + \d_\beta \eta^\alpha
\gamma_{\alpha\delta} + \d_\delta \eta^\alpha \gamma_{\alpha\beta} =
0,\\
\label{eq:suppleqBY}
g^{rr}\left( \bar N \d_r\xi - \xi \d_r N\right) = - g^{rA}\left( \bar N
  \d_A\xi - \xi \d_A \bar N\right)
\end{gather}
where $x^\alpha = (t, x^A)$. Equation (\ref{eq:liegammaBY}) is the
Killing equation for the metric $\gamma_{\alpha\beta}$. Once we
have selected a Killing vector $\eta^\alpha$, the last equation
(\ref{eq:suppleqBY}) gives $\d_r \xi$ in term of the other quantities.

\vspace{3mm}

The next step is to select the allowed transformation for which we
can build a differentiable generator. This computation is similar to
the one we did in order to compute the total Hamiltonian in section
\ref{sec:bound-cond-total}. As the boundary values of $\xi, \xi^i$ are
independent of the dynamical fields, we see easily that they are all
differentiable with a generator given by:
\begin{equation}
\label{eq:BYdiffgene}
\Gamma_\xi = \int_\Sigma \left( \xi \cR + \xi^i \cR_i\right) d^nx +
2\oint_{\d \Sigma} \left(\xi^A \pi^r_{\phantom r A} + \xi
  \sqrt{\gamma} K \right) d^{n-1}x.
\end{equation}

\vspace{5mm}

\begin{theorem}
To each Killing vector $\eta^\alpha$ of the induced metric $\gamma_{\alpha\beta}$ on
the boundary $\d \Sigma \times \RR$, we can associate a differentiable
generator $\Gamma_\xi$ given in (\ref{eq:BYdiffgene}) where the
boundary values of $\xi$ satisfy:
\begin{equation}
\xi \vert_{\d \Sigma} = \bar N \eta^0, \quad \xi^A \vert_{\d \Sigma}
= \eta^A + \bar N^A \eta^0, \quad \xi^r \vert_{\d \Sigma} = 0,
\end{equation}
and $\d_r\xi \vert_{\d \Sigma}$ is given by
(\ref{eq:suppleqBY}). Evaluated on the constraints surface, these
generators are conserved quantities:
\begin{equation}
\Gamma_\xi[g_{ij}, \pi^{ij}] \approx \oint_{\d \Sigma} \sqrt{\gamma}
\left(\eta^A \sigma_A + \eta^0
  \epsilon \right) d^{n-1}x,
\end{equation}
where the energy and momentum density are defined by:
\begin{equation}
\epsilon = 2 \bar N K + \frac{2}{\sqrt{\gamma}}\bar N^A
\pi^r_{\phantom r A}, \qquad \sigma_A = \frac{2}{\sqrt{\gamma}}
\pi^r_{\phantom r A}.
\end{equation}
\end{theorem}
\noindent This is exactly the result obtained in \cite{Brown1993}. A
straightforward computation with repeated use of equations
\eqref{eq:gaugegI} and \eqref{eq:gaugegII} leads to
\begin{eqnarray}
	\left\{ \Gamma_{\xi_1}, \Gamma_{\xi_2}\right\} & \approx &
	2 \oint_{\d\Sigma} \pi^r_A \left(\xi_1^B\d_B \xi_2^A -\xi_2^B\d_B \xi_1^A + \gamma^{AB}
	(\xi_1 \d_B \xi_2 -\xi_2 \d_B \xi_1 )\right)d^{n-1}x\nonumber\\
	&& + 2 \oint_{\d\Sigma} \sqrt{\gamma} K \left(\xi_1^A\d_A\xi_2 -
\xi_2^A\d_A\xi_1\right)d^{n-1}x \nonumber \\
&\approx& \oint_{\d\Sigma}\sqrt \gamma \sigma_A \left(\eta_1^\alpha\d_\alpha \eta_2^A - (1\leftrightarrow2)\right)d^{n-1}x\nonumber\\
	&& + \oint_{\d\Sigma} \sqrt{\gamma} \epsilon \left(\eta^\alpha_1
\d_\alpha \eta^0_2 - (1\leftrightarrow2)\right)d^{n-1}x,
\end{eqnarray}
which is the expected algebra on the boundary.

%%%%%%%%%%%%%%%%%%%%%%%%%%%%%%%%%%%%%%%%%%%%%%%%%%%%%%%%%%%%%%%%%%%%%%

\newpage
\section{Boundary Gauge Degrees of Freedom and Holography}
\label{sec:reduced-phase-space}

We saw in section \ref{sec:surface-charges} how the notion of differentiable
functional solves the problem of surface charges by selecting the
right generator for gauge symmetries. We will show in this section
how we can use the same notion to extract information about the reduced phase-space of some
theories with no local degrees of freedom without solving the
constraints. In those cases, the phase-space of boundary gauge degrees
of freedom is described by
dynamical fields living on the boundary of the spatial manifold $\Sigma$. This can be interpreted as direct examples of the holography mechanism. 
We will also obtain a
complete classification of the possible boundary conditions for the bulk
theory and the dictionary with the corresponding Hamiltonians of the
boundary theory.

\vspace{5mm}

In this section, we will consider field theories of the form:
\begin{equation}
\label{eq:SCactionRPS}
S[z^A, \lambda^\alpha] = \int dt \left[\int_\Sigma d^nx \left( \frac{1}{2}
  \sigma_{AB}z^A\dot z^B - \lambda^\alpha \gamma_{\alpha} \right)
+ \oint_{\d \Sigma} b\right],
\end{equation}
where $\gamma_\alpha$ are first-class constraints. We will assume that
a set of boundary conditions for $\lambda^\alpha$ and $z^A$ as been
selected such that the total Hamiltonian $H_T$ is a differentiable generator:
\begin{equation}
H_T = \int d^nx\, \lambda^\alpha \gamma_{\alpha}
+ \oint_{\d \Sigma} b.
\end{equation}
Using the results of section \ref{sec:poss-bound-cond}, we see that
once boundary conditions have been selected for $z^A$, the possible
total Hamiltonians are given by the set of differentiable gauge
generators. Computing the set of differentiable gauge transformations
$\delta_\eta$ gives the set of possible boundary conditions
on $\lambda^\alpha$.

\vspace{5mm}

We will restrict our analysis to theories with no local degrees of
freedom, i.e. where the number of independent first-class constraints is
equal to the number of canonical pairs. The reason is that if we have
local degrees of freedom, we need boundary conditions on the dynamical
fields $z^A$, see example \ref{firstclassEM}. This would considerably
restricts the set of possible improper gauge transformations. Some
examples of  
theories with no local degrees of freedom are Poisson sigma models in 2D, Chern-Simons theories,
pure gravity in 3 dimensions, BF theories, ... 

In the following, we will use Chern-Simons theories in 3D and BF
theory in 4D to present our technique. For simplicity, we will restrict our analysis
to finite manifold $\Sigma$. If $\Sigma$ has a boundary
at infinity, the analysis is similar with the additional requirement that all differentiable
generators have to be finite. The analysis of pure gravity in
3D will the subject of a following work. 

We will start by studying 3D Chern-Simons without imposing any
boundary conditions on the dynamical variables $A^a_i$. We will
compute the Dirac bracket using differentiable gauge generators and
make the link with the Wess-Zumino-Witten description. We will then
impose boundary conditions on the dynamical variables $A^a_i$ and show
how those new boundary conditions can be interpreted as boundary
constraints on the theory. This allows us to completely classify the
possible boundary conditions for Chern-Simons in 3D. In the last part,
we will apply the same analysis to the 4D BF theory. In particular, we
will recover the link between the reduced phase-space of 4D BF and the
phase-space of a 3D Chern-Simons theory defined on the boundary.

\subsection{3D Chern-Simons}
\label{sec:3d-chern-simons}

Chern-Simons theory in 3 dimensions is a good toy model. The constraints
can be solved exactly and one can show that the reduced phase-space
theory is given by a WZW model on the boundary \cite{Moore1989,Elitzur1989}. In this section, we
will recover the same result using differentiable gauge
transformations without having to solve the constraints. 

The hamiltonian bulk action is:
\begin{equation}
\label{eq:actionCS}
S[A^a_i, A^a_0] = \frac{-\kappa}{2\pi} \int dt \int_\Sigma d^2x \,
\frac{1}{2} \epsilon^{ij}g_{ab} \left(A^a_i \dot A^b_j - A^a_0 F^b_{ij} \right).
\end{equation}
We use $\epsilon^{12}=1$ and the metric $g_{ab}$ is a symmetric
non-degenerate invariant tensor on the Lie algebra $\mathfrak g$ of
the Lie group $\mathfrak G$. The fields
$A^a_i$ are the dynamical variables and $A^a_0$ plays the role of
lagrange multipliers. They are all valued in the algebra $\mathfrak g$:
$A_i\equiv A^a_iT_a \in \mathfrak g$ and $A_0\equiv A^a_0T_a \in \mathfrak g$ where $T_a$ are the generators of
$\mathfrak g$. We
also use the usual field strength in 2 
dimensions $F^a_{ij} = \d_iA^a_j - \d_j A^a_i +
f^a_{bc}A^b_iA^c_j$ where $f^a_{bc}$ are the structure constants of $\mathfrak g$. The poisson bracket and the 
constraints can be easily read from the action to be:
\begin{eqnarray}
\left\{F,G \right\} &=& \frac{2\pi}{\kappa}\int_\Sigma d^2x \, \frac{\delta F}{\delta
  A^a_i} \epsilon_{ij} g^{ab}\frac{\delta G}{\delta
  A^b_j}, \label{eq:CSPoisBrac} \\
\Phi_a & = & \frac{-\kappa}{4\pi} \, g_{ab} \epsilon^{ij}
F^b_{ij}\approx 0. \label{eq:CSConstraints}
\end{eqnarray}
The constraints satisfy the following closed algebra
\begin{equation}
\left\{\Phi_a(x),\Phi_b(y) \right\} = - f^c_{ab}\Phi_c(x)\,  \delta^2(x-y).
\end{equation}
Because the hamiltonian is a combination of the primary constraints
and these constraints satisfy to a closed algebra, we have the full
set of constraints and they are all first-class. 

For simplicity, we will consider $\Sigma$ to be a disk of finite
radius $R$ and we will use adapted coordinates $r, \phi$. There are
multiple choices for the boundary conditions, the most common being
$A_0\vert_{\d \Sigma}=0$ or $A_0 - A_\phi \vert_{\d \Sigma} =
0$. Those are boundary conditions on the lagrange multipliers only:
they don't impose anything on the dynamical fields. We can study
the canonical structure and the set of differentiable gauge
transformations imposing only the constraints and all their derivatives as
boundary conditions.

The gauge-like transformations are given by:
\begin{equation}
\delta_\eta A^a_i = \frac{-1}{2} g^{ab}\epsilon_{ij}
\frac{\delta}{\delta A^b_j}\left( \eta^c g_{cd}\epsilon^{kl}F^d_{kl} \right).
\end{equation}
The proper gauge transformations are those for which the gauge parameters
are zero on the boundary $\eta \vert_{\d \Sigma}=0$. On the
constraints surface, the finite gauge
transformations of $A_r$ are given by:
\begin{equation}
A'_r = h^{-1}A_r h + h^{-1}\d_r h,
\end{equation}
where $h$ is valued in the group $\mathfrak G$. The subset of finite proper gauge
transformations is the set of transformations for which $h$ is the identity on the
boundary. Using a finite proper gauge transformations, we can put $A^a_r$ to zero in a
neighborhood of the boundary but this
uses all the gauge freedom that we have in that neighborhood. The
suitable group element is given by:
\begin{equation}
h = \cP \exp\left( \int_r^R dr' A_r(r', \theta) \right),
\end{equation}
where $\cP$ is the path ordering symbol. With $A_r=0$, the constraints take the form
\begin{equation}
0 \approx \Phi_a = \frac{-\kappa}{4\pi} g_{ab}\left( \d_r A^b_\phi - \d_\phi
  A^b_r + f^b_{cd}A^c_rA^b_\phi\right) = \frac{-\kappa}{4\pi} g_{ab}\d_r A^b_\phi.
\end{equation}
The value of $A^a_\phi$ is completely characterized by its
boundary value $A^a_\phi\vert_{\d \Sigma}$. The gauge transformation
of $A^a_\phi$ is
\begin{equation}
\delta_\eta A^a_\phi = \d_\phi \eta^a + f^a_{bc}A^b_\phi \eta^c \equiv
D_\phi \eta^a.
\end{equation}
Evaluated on the boundary for a proper gauge transformation, we obtain
$\delta_\eta A^a_\phi\vert_{\d\Sigma} = 0$: this boundary value is a
gauge invariant quantity. It means that, up to topological
issues, the reduced phase-space of the theory is parametrized by the
boundary value of $A^a_\phi$.
We will now use the differentiable gauge
generators and their algebra to compute the induced bracket.

The boundary conditions on the dynamical fields $A^a_i$ are preserved by
gauge-like transformations: the
only restrictions from differentiability are coming from the existence of the
boundary term for the generator. The bulk generator is $\bar
\Gamma_\eta[A^a_i] = \int_\Sigma d^2x \, \eta^a \Phi_a$ and its
variation  is given by:
\begin{eqnarray}
\delta \bar \Gamma_\eta[A^a_i] &=& \int_\Sigma d^2x \, \delta\eta^a \Phi_a +
\frac{\kappa}{2\pi} \int_\Sigma d^2x \, D_i\eta^a \, g_{ab} \epsilon^{ij}
\delta A^b_{j} \nonumber \\ && \quad- \frac{\kappa}{2 \pi}\oint_{\d \Sigma} (d^1x)_i \,
\epsilon^{ij}\eta^ag_{ab} \delta A^b_j.
\end{eqnarray}
There might be boundary terms coming from $\delta \eta^a$ but they will
be proportional to $\Phi_a$ which is zero on the boundary.
When evaluated on the constraints surface, the boundary term is given
by the last term only. We are looking for a $(n-1,0)$-form $k_\eta$ such that:
\begin{equation}
\delta\oint_{\d \Sigma} k_\eta^{n-1} =\frac{\kappa}{2
  \pi}\oint_{\d \Sigma} d\phi \, \eta^ag_{ab} \, \delta A^b_\phi.
\end{equation}
The only gauge parameters for which we can build a boundary term are those
such that there exists $f_\eta$ a local function of $\phi$ and $A_\phi$ defined
on the boundary $\d \Sigma$ with:
\begin{equation}
\label{eq:CSboundcondgaugegen}
\eta^a \vert_{\d \Sigma} = \frac{2 \pi}{\kappa}g^{ab}\left.\sum^\infty_{k=0} (-\d_\phi)^k
\frac{\partial^S f_\eta}{\partial \d^k_\phi A^b_\phi}\right\vert_{\d \Sigma}\equiv  \left.\frac{2 \pi}{\kappa}g^{ab}
\frac{\bar\delta f_\eta}{\delta A^b_\phi}\right\vert_{\d \Sigma}.
\end{equation}
We have defined $\frac{\bar \delta}{\delta A^a_\phi}$ as the Euler-Lagrange
derivative on the circle $\d \Sigma$.
The differentiable generators of the gauge transformations are then given by:
\begin{equation}
\Gamma_\eta[A^a_i] = \int_{\Sigma} d^2x \, \eta^a \Phi_a + \oint_{\d
  \Sigma} d\phi \, f_\eta \approx \oint_{\d
  \Sigma} d\phi \, f_\eta .
\end{equation}
This is the complete set of differentiable generators of gauge
transformations.

We showed that to any functional $F[A_\phi]=\int_{\d \Sigma} d\phi f$ on the circle $\d\Sigma$, we can
associate a differentiable gauge generator $\Gamma_F[A^a_i]$ by
choosing a gauge parameter $\eta$ satisfying
\eqref{eq:CSboundcondgaugegen}. There are multiple possible generators
but they are all equals on the constraint surface:
\begin{equation}
\Gamma_F[A^a_i] \approx F.
\end{equation}
As differentiable gauge generators are first-class quantities, we can
compute their Dirac bracket by computing their Poisson bracket and
evaluate the result on the constraints surface. If one considers two functionals
of $A^a_\phi$ on the circle $\d \Sigma$, $F_1$ and $F_2$, the Poisson
bracket of the associated differentiable gauge generators
$\Gamma_{F_1}$ and $\Gamma_{F_2}$ is easily computed. We obtain:
\begin{eqnarray}
\left\{\Gamma_{F_1}, \Gamma_{F_2} \right\} & \approx & \frac{2\pi}{\kappa}
\int_\Sigma d^2x \, \left(\frac{\kappa}{2\pi} \epsilon^{ik}g_{ac}D_k\eta_1^c\right)
\,\epsilon_{ij} g^{ab} \, \left(\frac{\kappa}{2\pi}
  \epsilon^{jl}g_{bd}D_l \eta_2^d\right)\nonumber \\
&\approx & \frac{\kappa}{2\pi} \oint_{\d \Sigma} (d^1x)_i \,
\epsilon^{ij}\left(g_{ab} \eta^a_1 \d_j \eta_2^b - A^a_j
  g_{ab}f^b_{cd}\eta^c_1\eta^d_2\right).
\end{eqnarray}
We can then read the Dirac bracket induced on the functionals $F_1$ and $F_2$:
\begin{equation}
\label{eq:CSDiracFunct}
\left\{ F_1, F_2\right\}^* \approx \frac{2\pi}{\kappa} \oint_{\d \Sigma}
d\phi\left(g^{ab} 
  \frac{\bar \delta f_1}{\delta A^a_\phi} \d_\phi \frac{\bar \delta f_2}{\delta A^b_\phi} - A^a_\phi
  g_{ab}f^{bcd}\frac{\bar \delta f_1}{\delta A^c_\phi}\frac{\bar \delta f_2}{\delta A^d_\phi}\right).
\end{equation}
This can also be written in term of $A^a_\phi\vert_{\d \Sigma}$:
\begin{equation}
\label{eq:CSDiracAphi}
\left\{ A^a_\phi\vert_{\d \Sigma}(\phi), A^b_\phi\vert_{\d \Sigma}(\phi')\right\}^* \approx
\frac{2\pi}{\kappa} \left(g^{ab} \d_\phi- 
  g_{cd}f^{dab}A^c_\phi\vert_{\d \Sigma}\right) \delta (\phi-\phi').
\end{equation}
The fields $A^a_\phi\vert_{\d \Sigma}$ parametrize the reduced
phase-space of the theory. Their Dirac bracket is given by
\eqref{eq:CSDiracFunct} and \eqref{eq:CSDiracAphi}. The only thing we
are missing to have the full description is the Hamiltonian. 

The total
Hamiltonian $H_T$ is the differentiable gauge generator associated to
the lagrange multipliers $A_0$. For this generator to exists, $A_0$ must
satisfy boundary conditions of the form
\eqref{eq:CSboundcondgaugegen} for a particular functional $\oint_{\d \Sigma} d\phi \,
h_B$ of $A^a_\phi$:
\begin{equation}
A^a_0\vert_{\d \Sigma} = \left.\frac{2 \pi}{\kappa} g^{ab}\frac{\bar \delta
  h_B}{\delta A^b_\phi}\right \vert_{\d \Sigma}.
\end{equation}
The total Hamiltonian is then given by:
\begin{equation}
H_T = \int_\Sigma A^a_0 \Phi_a + \oint_{\d \Sigma} d\phi \, h_B. 
\end{equation}
The Hamiltonian on the reduced phase-space is the value of $H_T$
evaluated on the constraints surface:
\begin{equation}
H_T \approx \oint_{\d \Sigma} d\phi \,h_B. 
\end{equation}
As we said earlier, the two most common boundary conditions for
$A^a_0$ are $A^a_0=0$ or $A^a_0=A^a_\phi$ on $\d
\Sigma$. Respectively, they correspond to $h_B=0$ and $h_B =
\frac{\kappa}{4\pi}g_{ab}A^a_\phi A^b_\phi$. 

\vspace{3mm}

The global picture is the following: up to topological issues, the reduced
phase-space is parametrized by the value of $A^a_\phi$ on the boundary
and the Dirac bracket is given by:
\begin{equation}
\label{eq:CSDiracAphiII}
\left\{ A^a_\phi\vert_{\d \Sigma}(\phi), A^b_\phi\vert_{\d \Sigma}(\phi')\right\}^* =
\frac{2\pi}{\kappa} \left(g^{ab} \d_\phi- 
  g_{cd}f^{dab}A^c_\phi\vert_{\d \Sigma}\right) \delta (\phi-\phi').
\end{equation}
This is the current algebra associated to the algebra $\mathfrak g$
and we recovered the result obtained in \cite{Elitzur1989}. 
The Hamiltonian is given by a functional of $A^a_\phi$:
\begin{equation}
\label{eq:CSReducedHamil}
H_T[A^a_\phi\vert_{\d \Sigma}] \approx \oint_{\d \Sigma} d\phi \, h_B(A^a_\phi)
\end{equation}
which is determined by the
boundary conditions on the lagrange multipliers $A^a_0$ through:
\begin{equation}
A^a_0\vert_{\d \Sigma} = \left.\frac{2 \pi}{\kappa} g^{ab}\frac{\bar \delta
  h_B}{\delta A^b_\phi}\right \vert_{\d \Sigma}.
\end{equation}

If we consider the problem from the other direction, we see that only
the reduced phase-space structure is controlled by the bulk
action. The Hamiltonian \eqref{eq:CSReducedHamil} is determined by the
boundary conditions of $A^a_0$ and by tuning them, we can build any local
function $h_B$. It has important consequences. For instance, it
implies that any theory in two dimension with the 
phase-space structure given by \eqref{eq:CSDiracAphiII} is equivalent
to a 3D Chern-Simons theory with specific boundary conditions. It
also means that two different choices for the boundary conditions of
$A^a_0$ will lead to different Hamiltonians on the reduced
phase-space: they describe two completely different theories. 

\vspace{5mm}

The analysis has been done strictly at the level of the
boundary terms without taking into account the topological structure
of $\Sigma$. This structure will in general restrict the set of
available values for $A_\phi$ on the boundary. In the case where
$\Sigma$ is a disk, the reduced phase-space is
exactly parametrized by the value of $A_\phi$ on the boundary such that the
holonomy around the $\phi$-circle is the identity of the group
$\mathfrak G$:
\begin{equation}
W \equiv \cP\exp\left(-\int_0^{2 \pi} d\phi \, A_\phi\vert_{\d
    \Sigma} \right) = 1 \in \mathfrak G.
\end{equation}

\subsection{Boundary conditions on the phase-space}
\label{sec:BC_on_the_phase-space}

In the previous section, we considered only boundary conditions on the
Lagrange multipliers. However, one can construct sets of boundary
conditions that also include boundary conditions on the canonical
variables $A^a_\phi$. For instance, boundary conditions on $A^a_\phi$
are present in the study of Chern-Simons gravity
\cite{Coussaert1995,Henneaux2000}. Those additional boundary
conditions are responsible for the second step of the reduction from
Chern-Simons to Liouville in the Brown-Henneaux boundary conditions
case: first, one goes to WZW on the boundary and then, there is a
second reduction to Liouville. 

\vspace{3mm}

Let's assume that we want additional boundary conditions on the canonical variables:
\begin{equation}
\label{eq:condaphi}
\chi^\alpha(A^a_\phi)\vert_{\d \Sigma} = 0 
\end{equation}
where $\chi^\alpha$ are local functions of $A^a_\phi$ and their
derivatives $\d_\phi^k A^a_\phi$. Any boundary condition on $A^a_r$ is
irrelevant because $A^a_r$ is pure gauge even on the boundary. The
conservation of $\chi^\alpha = 0$ in time will impose restrictions on
the possible boundary conditions for $A^a_0$: 
\begin{equation}
\left.\sum_k \frac{\partial \chi^\alpha}{\partial \d^k_\phi A^a_\phi} \d_k (D_\phi A^a_0) \right\vert_{\d \Sigma}= 0. 
\end{equation}
By adding test functions $\nu_\alpha$ on the circle, we can rewrite this as an integral:
\begin{equation}
\label{eq:conscondaphi}
\oint_{\d \Sigma} d\phi \, \frac{\bar \delta (\nu_\alpha \chi^\alpha)}{\delta A^a_\phi} D_\phi A^a_0 = 0 \qquad \forall \nu_\alpha.
\end{equation}
Let's assume that we have a set of boundary conditions on $A^a_0$
satisfying \eqref{eq:conscondaphi}
such that the differentiable total Hamiltonian $H_T$ can be
constructed. Remark thtat the differentiability condition here is different than
the one used in section \ref{sec:3d-chern-simons} as we now have extra boundary conditions on
the dynamical variables. 

We want to relate the two problems and describe the case with the additional conditions
$\chi^\alpha$ in term of the canonical structure of section \ref{sec:3d-chern-simons}. The idea will be to treat those additional
conditions as additional constraints. Imposing everything, $F_{ij}=0$
and $\chi^\alpha\vert_{\d \Sigma}=0$, the total hamiltonian $H_T$
becomes a functional of $A^a_\phi\vert_{\d \Sigma}$:
\begin{equation}
H_T\vert_{\chi^\alpha = 0} \approx \oint_{\d \Sigma}d\phi \, h_B (A^a_\phi).
\end{equation}
We saw in the previous section that there exist a differentiable
generator $\tilde H_T$ for the canonical structure without imposing
$\chi^\alpha$ such that 
\begin{equation}
\tilde H_T \approx \oint_{\d \Sigma}d\phi \, h_B (A^a_\phi).
\end{equation}
By construction, if we impose the additional boundary conditions
\eqref{eq:condaphi}, we have: 
\begin{equation}
H_T\vert_{\chi^\alpha = 0} = \tilde H_T\vert_{\chi^\alpha = 0},
\end{equation}
which implies that $\tilde H_T$ is differentiable using either of the
two canonical structures.
From here on, we will work with the general canonical structure of
section \ref{sec:3d-chern-simons} and the Hamiltonian will be taken as $\tilde
H_T$. For any function $\nu_\alpha$ on the circle, we can build the 
associated differentiable gauge generator $\Gamma_\nu$ such that: 
\begin{equation}
\Gamma_\nu[A^a_i] \approx \oint_{\d \Sigma} d\phi \, \nu_\alpha \chi^\alpha.
\end{equation}
The well-defined action for the total system can then be written as:
\begin{equation}
S[A^a_i, A^a_0 = \bar A^a_0 + \mu^a, \nu_\alpha] = \frac{-\kappa}{2\pi} \int dt \left\{\int_\Sigma d^2x \,
\frac{1}{2} \epsilon^{ij}g_{ab}A^a_i \dot A^b_j - \tilde H_T - \Gamma_\nu\right\},\nonumber
\end{equation}
where $\mu^a$ with $\mu^a\vert_{\d \Sigma}=0$ are the true lagrange
multipliers enforcing $\Phi_a$. 
The action is written without imposing a priori the conditions
$\chi^\alpha\vert{\d\Sigma}=0$; they will be enforced by the boundary
term coming from the variation of $\nu$ in $\Gamma_\nu$. Solving the
constraints $\chi^\alpha=0$ reduces the action to:
\begin{equation}
S[A^a_i, A^a_0 = \bar A^a_0 + \mu^a] = \frac{-\kappa}{2\pi} \int dt \left\{\int_\Sigma d^2x \,
\frac{1}{2} \epsilon^{ij}g_{ab}A^a_i \dot A^b_j - H_T\right\}.
\end{equation}
This is the expected action when $\chi^\alpha\vert_{\d \Sigma}$ hold.
On the reduced phase-space, we obtain:
\begin{equation}
\tilde H_T + \Gamma_\nu \approx \oint_{\d\Sigma} d\phi \, \left(h_B + \nu_\alpha \chi^\alpha \right),
\end{equation}
which is the Hamiltonian of a constrained system. The condition \eqref{eq:conscondaphi} can be rewritten
\begin{equation}
\left\{\oint_{\d\Sigma} d\phi \, h_B, \oint_{\d\Sigma} d\phi \, \nu_\alpha\chi^\alpha \right\}^*=0 \qquad \forall \nu_\alpha.
\end{equation}
As expected, this condition is the conservation of the constraints
$\chi^\alpha$ under time evolution. 

We see explicitly that boundary conditions on the canonical variables
of the Chern-Simons theory will produce a constrained
Wess-Zumino-Witten model on the boundary. In the case of gravity,
solving those additional constraints and going to the fully reduced
phase-space is what gives Liouville theory. 

\vspace{5mm}

The other side of the same coin is that, playing with the boundary
conditions of the Lagrange multipliers $A^a_0$, we can build any
Hamiltonian on the reduced phase-space. In particular, we can build a Hamiltonian containing some
additional constraints on $A^a_\phi \vert_{\d \Sigma}$. 

Let's assume that we want a Hamiltonian containing constraints on our boundary. We select $h_B$ to be of the form:
\begin{equation}
h_B = \tilde h_B + \nu_\alpha \chi^\alpha,
\end{equation}
where $\nu_\alpha$ are test function of the boundary that can vary arbitrarily. The associated boundary conditions for $A^a_0$ are given by:
\begin{equation}
A^a_0 \vert_{\d \Sigma} = \frac{2\pi}{\kappa} g^{ab} \frac{\bar \delta }{\delta A^b_\phi} \left( \tilde h_B + \nu_\alpha \chi^\alpha \right).
\end{equation}
Because $\nu_\alpha$ is not fixed, part of the $A^a_0$ are not fixed on the boundary but are in reality playing the role of the Lagrange multipliers enforcing the additional boundary constraints $\chi^\alpha$. This is particularly obvious if the boundary constraints $\chi^\alpha$ are simple boundary conditions like:
\begin{equation}
\label{eq:simplbondcondaphi}
A^{\bar a}_\phi \vert_{\d \Sigma} = 0,
\end{equation}
where $\bar a$ is fixed. In that case, the combination $g_{\bar a
  b}A^b_0$ plays the role of the Lagrange multiplier enforcing
\eqref{eq:simplbondcondaphi}.

\vspace{5mm}

As the analysis we did on one boundary can be reproduced independently
on all boundaries, 
we obtained a complete classification of the possible boundary
conditions for Chern-Simons theory in 3 dimensions:
\begin{theorem}
For Chern-Simons theory on $\Sigma \times \RR$, the
possible boundary conditions for the fields $A_\mu^a$ on each connected
component $\cC_n$ of $\d \Sigma$ are in one to one correspondance with
the functionals of $a_n^a$ defined on the circle
$\cC_n$ where the fields $a_n^a$ are the pullback of $A_i^a$ on $\cC_n$.
\end{theorem}

\subsection{4D BF Theory}
\label{sec:bf-theory}

In this section, we will study BF theory in 4 dimensions with a
cosmological term \cite{Horowitz1989, Cattaneo,
  Baez1996}. The bulk term of the hamiltonian action can be written
as: 
\begin{eqnarray}
S[A^a_i, B^a_{ij}, A^a_0, B^a_{0i}] & = & \int dt\int_\Sigma d^3x \, \left\{ B^a_{ij} \epsilon^{ijk} g_{ab} \d_t A^b_k -  A^a_0 \Phi_a - B^a_{0i} \Psi^i \right\},\\
\Phi_a & = & -  g_{ab}\epsilon^{ijk} D_i B^b_{jk},\\
\Psi^i_a & = & -  g_{ab}\epsilon^{ijk} \left(F^b_{jk} + \frac{\Lambda}{6}B^b_{jk}\right). 
\end{eqnarray}
All fields are valued in the algebra $\mathfrak g$ and we use the same
convension as in section \ref{sec:3d-chern-simons}.
The Poisson bracket is given by:
\begin{equation}
\left\{ I, J \right\} = \frac{1}{2}\int_\Sigma d^3x \, \left( \frac{\delta I}{\delta A^a_i}g^{ab} \epsilon_{ijk}\frac{\delta J}{\delta B^b_{jk}} - I \leftrightarrow J \right).
\end{equation}
The constraints are first-class and satisfy the following algebra:
\begin{eqnarray}
\left\{ \Phi_a(x), \Phi_b(y)\right\} &= & - f^c_{ab} \Phi_c \delta^3(x-y), \\
\left\{ \Phi_a(x), \Psi_b^i(y)\right\} &= & - f^c_{ab} \Psi_c^i \delta^3(x-y), \\
\left\{ \Psi_a^i(x), \Psi_b^j(y)\right\} & = & 0.
\end{eqnarray}
We have $3N$ canonical pairs and $4N$ first-class constraints. The naive counting leads to $-N$ local degree of freedom. However, locally, the constraints are not independent:
\begin{equation}
D_i \Psi^i_a = \frac{\Lambda}{6}\Phi_a.
\end{equation}
There are only $3N$ locally independent constraints and, as expected, zero local degrees of freedom. On the boundary theory however, the story will be different.

The setup is very similar to the 3D Chern-Simons theory described in
the previous sections and we don't need any boundary conditions on the canonical variables to make the action well-defined. The only boundary conditions that we need are on the Lagrange multipliers and will give the Hamiltonian of the reduced theory. Let's first compute the reduced phase-space. As before, we will focus on one boundary and ignore possible topological obstructions. 

\vspace{3mm}

The smeared constraints will be denoted
\begin{equation}
\bar \Gamma_{\epsilon, \eta} = \int_\Sigma d^3x \, \left( \epsilon^a \Phi_a + \eta^a_i \Psi^i_a \right).
\end{equation}
They are associated to the following gauge-like transformations:
\begin{eqnarray}
\delta_{\epsilon, \eta} A^a_i & = & D_i \epsilon^a - \frac{\Lambda}{6}\eta_i^a, \\
\delta_{\epsilon, \eta} B^a_{ij} & = & D_i \eta^a_j - D_j \eta^a_i + f^a_{bc}B^b_{ij} \epsilon^c.
\end{eqnarray}
The boundary term in the variation of $\bar \Gamma_{\epsilon, \eta}$ is
\begin{gather}
\oint_{\d \Sigma} I^n (\bar \Gamma_{\epsilon, \eta}) =- \oint_{\d
  \Sigma} (d^2x)_i \, \epsilon^{ijk} g_{ab} \left( \epsilon^a \delta
  B^b_{jk} - 2 \eta^a_j \delta A^b_k \right),\\
(d^2x)_i = \frac{1}{2}\epsilon_{ijk}dx^jdx^k.
\end{gather}
As before, let's use coordinates adapted to the boundary $x^i = (r,
x^A)$ with the boundary under consideration given by $r$ constant. In
that case, the gauge-like transformations with $\epsilon^a=0$ and
$\eta^a_A=0$ on the boundary are proper gauge transformations. On the
constraints surface, the
finite gauge transformations are generated by the two following
transformations:
\begin{equation}
\label{eq:BFfiniteI}
\left\{ \begin{array}{rcl}
A'_i &=& h^{-1} A_i h + h^{-1}\d_i h,\\
B'_{ij} & = & h^{-1} B_{ij} h,
  \end{array}\right.  \qquad h \in \mathfrak G,
\end{equation}
and
\begin{equation}
\label{eq:BFfiniteII}
\left\{ \begin{array}{rcl}
A'_i &=& A_i -\frac{\Lambda}{6} \eta_i,\\
B'_{ij} & = & B_{ij} +D_i \eta_j - D_j \eta_i - \frac{\Lambda}{6}
[\eta_i, \eta_j],
  \end{array}\right. \qquad \eta_i \in \mathfrak g,
\end{equation}
where $B_{ij}=B^a_{ij}T_a$. The finite proper gauge transformations are those
generated by transformations with $h$ equals to the identity and $\eta_i$ equals to zero on the
boundary. Using a proper transformation of the form \eqref{eq:BFfiniteI}, we can
put $A_r=0$ in a neighborhood of the boundary (see section
\ref{sec:3d-chern-simons}). We can then use a transformation of the
form \eqref{eq:BFfiniteII}, with $\eta_r=0$ and $\eta_A$ solution to
\begin{equation}
\d_r \eta_A = -B_{rA}, \qquad \eta_A\vert_{\d \Sigma}=0,
\end{equation}
to also put $B_{rA}=0$ in a neighborhood of the boundary.
This fixes the gauge close to the boundary and the
reduce-phase space is then completely parametrized by the boundary
value of $A^a_A$ and $B^a\equiv\epsilon^{AB}B^a_{AB}$ with $\epsilon^{AB} = \epsilon^{rAB}$. However, they are not independent:  
\begin{equation}
\Psi^r_a = g_{ab}\left(\frac{\Lambda}{6}B^b + \epsilon^{AB}F^b_{AB}\right)\approx 0.
\end{equation}
The 4 sets of constraints
are dependent in the bulk but we see that on the boundary it is not
the case: imposing $\Phi_a\approx 0$ and $\Psi_a^A\approx 0$ imply $D_r
\Psi_a^r\approx 0$ but we still need to impose $\Psi_a^r\approx 0$ on
the boundary for it to be valid everywhere. 

The Dirac bracket is easily computed using the differentiable functionals $\Gamma_{F}$
\begin{gather}
F[A^a_A\vert_{\d \Sigma}, B^a\vert_{\d \Sigma}] = \oint_{\d \Sigma} d^2x \, f(A^a_A\vert_{\d \Sigma}, B^a\vert_{\d \Sigma}),\\
\Gamma_{F} = \Gamma_{\epsilon_F, \eta_F} + F[A^a_A\vert_{\d \Sigma}, B^a\vert_{\d \Sigma}], \\
\eta^a_{FA}\vert_{\d \Sigma} = - \epsilon_{AB} \frac{g^{ab}}{2}
\frac{\bar \delta f}{\delta A^b_B},\qquad 
\epsilon^a_{F}\vert_{\d \Sigma} = g^{ab} \frac{\bar \delta f}{\delta B^b}, 
\end{gather}
where the Euler-Lagrange
derivatives $\frac{\bar \delta}{\delta}$ are defined on the
boundary coordinates only. For two arbitrary functionals
$I[A^a_A\vert_{\d \Sigma}, B^a\vert_{\d \Sigma}]$ and
$J[A^a_A\vert_{\d \Sigma}, B^a\vert_{\d \Sigma}]$, a direct computation leads to:
\begin{eqnarray}
\left\{ I, J\right\}^* & \approx & \left \{ \Gamma_{I}, \Gamma_{J}\right\}\nonumber\\
 & \approx & \oint_{\d \Sigma} d^2x \,g_{ab}
 \,\left(-B^af^b_{cd}\epsilon^c_I \epsilon^d_J-2  \epsilon^{AB}D_A
   \epsilon^a_I  \eta^b_{JB}\right. \nonumber \\ && \qquad \qquad \qquad \left. +2  \epsilon^{AB}D_A
   \epsilon^a_J  \eta^b_{IB}+
  \epsilon^{AB}\frac{\Lambda}{3} \eta^a_{IA} \eta^b_{JB}\right)\nonumber\\
 & \approx & \oint_{\d \Sigma} d^2x \, \left( -B^a g_{ab}
   f^{bcd}\frac{\bar \delta I}{ \delta B^c}\frac{\bar \delta J}{
     \delta B^d} - g^{ab} D_A \frac{\bar \delta I}{\delta B^a}  \frac{\bar \delta J}{\delta A^b_A} \right. \nonumber \\ && \qquad \qquad \qquad
 \left. +g^{ab} D_A \frac{\bar \delta J}{\delta B^a}  \frac{\bar \delta I}{\delta A^b_A} +\frac{\Lambda}{12} \frac{\bar \delta I}{ \delta
     A^a_A}g^{ab}\epsilon_{AB} \frac{\bar \delta J}{ \delta
     A^b_B}\right). 
\end{eqnarray}
The residual constraints on the boundary
$\Psi^r_a\vert_{\d \Sigma}$ are Casimir functions of the Dirac bracket: 
\begin{equation}
\left\{\Psi^r_a\vert_{\d \Sigma} ,J\right\}^*\approx 0,
\end{equation}
for all functional $J$. 

The Hamiltonian $H_T$ will be the
differentiable gauge generator associated to the lagrange multipliers:
\begin{equation}
H_T = \Gamma_{A_0, B_{0i}}.
\end{equation}
As in the case of Chern-Simons, by
tuning the boundary conditions on the Lagrange multipliers, we can
build $H_T$ to be any functional of the reduced phase-space.

\begin{theorem}
For BF theory in 4D on $\Sigma \times \RR$, the
possible boundary conditions for the fields $A_\mu^a, B^a_{\mu\nu}$ on each connected
component $\cC_n$ of $\d \Sigma$ are in one to one correspondance with
the functionals of $a_{nA}^a$ and $b_{nAB}^A$ defined on the circle
$\cC_n$. The fields $a_{nA}^a$ and $b_{nAB}^a$ are respectively the
pullback of $A_i^a$ and $B_{ij}^a$ on $\cC_n$ and satisfy the
constraint implied by the pullback of $\epsilon_{ijk}\Psi^k_a \approx 0$.
\end{theorem}

\vspace{5mm}

If $\Lambda$ is different than zero, we can solve $\Psi^r_a\vert_{\d
  \Sigma}\approx0$ exactly with 
\begin{equation}
B^b\vert_{\d \Sigma} =-\frac {6}{\Lambda} \epsilon^{AB}F^b_{AB}\vert_{\d \Sigma},
\end{equation}
and describe the reduced phase-space in
term of $A^a_A\vert_{\d\Sigma}$. The Dirac bracket becomes:
\begin{equation}
\left\{ I[A^a_A\vert_{\d\Sigma}], J[A^a_A\vert_{\d\Sigma}]\right\}^*
\approx \frac{\Lambda}{12} \oint_{\d \Sigma} d^2x \, \frac{\bar
    \delta I}{ \delta 
     A^a_A}g^{ab}\epsilon_{AB} \frac{\bar \delta J}{ \delta
     A^b_B}. 
\end{equation}
This is exactly the Poisson bracket \eqref{eq:CSPoisBrac} of the
Chern-Simons theory in 3 dimensions. In this case, the reduced phase-space of the 4D
B-F theory is the phase-space of Chern-Simons theory in 3
dimensions. However, the Hamiltonians will in general be different.

If we can build any Hamiltonian, in principle we should be able to
reproduce the one of Chern-Simons. The easiest way of constructing it
is to add a boundary condition on the canonical variable $B^a_{AB}$:
$B^a_{AB}\vert_{\d \Sigma} =0$. Following the arguments of the
previous section, this can be done by relaxing the boundary condition
on the corresponding Lagrange multiplier: $A^a_0$. If we put the other
two relevant Lagrange multipliers, $B^a_{0A}$, to zero on the
boundary, the associated differentiable total Hamiltonian is given by:
\begin{gather}
\left.B^a_{0A}\right\vert_{\d \Sigma} =0, \\
H_T = \int_\Sigma d^3x \, \left(A^a_0 \Phi_a + B^a_{0i}
  \Psi_a^i\right) + \oint_{\d \Sigma} d^2x \,  g_{ab} \epsilon^{AB}
A^a_0 B^b_{AB}. 
\end{gather}
On the constraint surface $\Phi_a \approx 0$ and $\Psi^i_a \approx 0$, we obtain:
\begin{equation}
H_T \approx  \frac{-6}{\Lambda} \oint_{\d \Sigma}d^2x \,  A^a_0 \, g_{ab} \epsilon^{AB} F^b_{AB},
\end{equation}
which is the Hamiltonian of the 3D Chern-Simons theory \eqref{eq:actionCS}. 

%%%%%%%%%%%%%%%%%%%%%%%%%%%%%%%%%%%%%%%%%%%%%%%%%%%%%%%%%%%%%%%%%%%%%%

\newpage

\section{Conclusions}
\label{sec:conclusions}

Starting with the notion of symplectic structure and requiring field
theories to behave like discrete mechanical system naturally introduces
the notion of differentiable generators. Restricting the set of
functionals to the subset of differentiable one is then mandatory for the
definition of a Poisson bracket. With these definitions, the
hamiltonian structure of field theories behaves exactly like the one
of discrete mechanical systems. 

In the context of gauge theories, we showed that boundary conditions split into two categories: the boundary
conditions on the dynamical variables are part of the definition of the
canonical structure whereas the boundary conditions on the lagrange
multipliers are part of the choice of the Hamiltonian of the
theory.  The canonical structure leads to the definition of
differentiable gauge transformations as the gauge-like transformations
generated by differentiable functionals. We gave a complete classification of the possible boundary
conditions on the lagrange multipliers in term of these differentiable gauge
transformations. In theories with local degrees of freedom, we need
boundary conditions on the dynamical variables in order to control the
flux of radiation. This restricts the set of differentiable
generators a lot and consequently the set of possible boundary
conditions on the lagrange multipliers. We also showed how the restriction to the differentiable gauge
transformations preserving the boundary conditions
on the lagrange multipliers leads to the usual notion of surface charges.

\vspace{5mm}

In theories with no
local degrees of freedom, one can often remove the boundary conditions
on the dynamical variable which makes the set of differentiable gauge
transformations a lot bigger. This leads to two interesting consequences.
Firstly, using differentiable gauge generators, we can probe the reduced
phase-space and compute the Dirac bracket without solving the
constraints. Secondly, by tuning the
boundary conditions on the lagrange multipliers, we can construct any
hamiltonian on the reduced phase-space. We used 3D Chern-Simons and 4D
BF theory as example. In particular, we derived the complete
set of possible boundary conditions for these theories when defined on
manifold with time-like boundaries located at a finite distance.

\vspace{5mm}

The reduced phase-space of topological theories like Chern-Simons
contains boundary gauge degrees of
freedom. However, they
are not a feature of topological theories and we expect them to exist
in any gauge theory. In this paper, we used the differentiable
generator of gauge transformations to describe them. Unfortunately,
this technique is not generalizable to theories with local degrees of
freedom like gravity in 4 dimensions. As we saw, the problem comes from
the necessity of boundary conditions on the dynamical variables. In the future, it would be interesting to generalize the notion of
canonical structure presented in this paper in order to relax these
kind of boundary conditions. The hope would then be that the boundary
gauge degrees of freedom would again be described by gauge
generators. This would give some new insight on the notion of
holography for more general cases.

\section*{Acknowledgements}
\label{sec:acknowledgements}

%\addcontentsline{toc}{section}{Acknowledgments}

I would like to thank G.~Barnich, F.~Canfora, G.~Comp\`ere,
H.~Gonz\'alez, B.~Oblak, A.~Perez, P.~Ritter, D.~Tempo,
R.~Troncoso and J.~Zanelli for useful discussions. This work is
founded by the fundecyt postdoctoral grant 3140125. The Centro de Estudios Cient\'ificos
(CECs) is funded by the Chilean Government through the Centers of
Excellence Base Financing Program of Conicyt.

%%%%%%%%%%%%%%%%%%%%%%%%%%%%%%%%%%%%%%%%%%%%%%%%%%%%%%%%%%%%%%%%%%%%%%

\newpage
\appendix

\section{The Phase-Space}
\label{sec:phase-space}

We will consider the space manifold $\Sigma$ to be of dimension $n$ and
described by coordinates $x^i$. Its exterior derivative will be
denoted $d$ and will be treated as a Grassmann odd quantity: $dx^i
dx^j = - dx^j dx^i$.

The dynamical fields of the theory will be denoted $z^A$. The
phase-space of the theory is the set of allowed configurations:
\begin{equation}
\cF = \left\{ z^A(x), x^i \in \Sigma; \chi^\mu(z)\vert_{\d \Sigma} = 0\right\}.
\end{equation}
The conditions $\chi^\mu\vert_{\d \Sigma} = 0$ are the set of boundary
conditions imposed on the fields; this set may be empty. If the
boundary of $\Sigma$ is at infinity, boundary conditions are replaced
by asymptotic conditions. We will assume that the boundary conditions
are imposed on all equalities.

We will now describe the differential structure of $\cF$. We will
start by ignoring the boundary conditions and then describe the
implications they have on the general structure.

\subsection{Differential Structure}
\label{sec:diff-struct}

The exterior differential associated to the infinite dimensional
manifold $\cF$ will be denoted $\delta$. We will also treat it as a
Grassmann odd quantity, $\delta z^A \delta z^B = -\delta z^B \delta
z^A$, $\delta x^i=0$, and assume that it anti-commutes with the base
manifold differential: $\{d, \delta\}=0$. A general form will have
components in both directions. A $(p,q)$-form will be a $p$-form over
$\Sigma$ and a $q$-form over $\cF$.

A vector field $Q^A \frac{\d}{\d z^A}$ over $\cF$ is called an
evolutionary vector field with caracteristic $Q^A$. It represents a variation of the
fields by an amount $Q^A(z, x)$. The operator measuring this variation
is  Grassmann even and denoted $\delta_Q$. It satisfies $[\delta_Q,
\delta] = 0$, $[\delta_Q, d]=0$ and $[\delta_q, \d_i]=0$ where $\d_i$
is the total derivative with respect to $x^i$. The algebra of
evolutionary vector fields is given by:
\begin{equation}
[\delta_{Q_1}, \delta_{Q_2}] = \delta_{[Q_1, Q_2]} \qquad \text{with}
\qquad [Q_1, Q_2]^A = \delta_{Q_1}Q^A_2 - \delta_{Q_2}Q^A_1.
\end{equation}
The interior product $\iota_Q$ between an evolutionary vector field $Q^A
\frac{\d}{\d z^A}$ and a general $(p, q)$-form $\theta^{p,q}$ is given
by:
\begin{multline}
\iota_Q \theta^{p,q}(\delta z^{A_1}, ..., \delta z^{A_q}, dx^{i_1},
..., dx^{i_p}) =\\ \sum_{k=1}^q (-)^k\theta^{p,q}(\delta z^{A_1}, ..., \delta_Q z^{A_k}, ..., \delta z^{A_q}, dx^{i_1},
..., dx^{i_p}).
\end{multline}
We have
\begin{equation}
\iota_Q \delta + \delta \iota_Q = \delta_Q, \quad \iota_{Q_1}
\delta_{Q_2} +\delta_{Q_2} \iota_{Q_1} = \iota_{[Q_1, Q_2]}.
\end{equation}

In our analysis, we will work with functionals and their differentials
under $\delta$. A functional $F$ is defined as the integral of a
$(n,0)$-form:
\begin{equation}
F[z] = \int_\Sigma f d^nx.
\end{equation}
We will use lowercase letters for the integrant and uppercase letters
for integrated quantities. If the integrant is a $(n,s)$-form
$\theta^{n,s}$, the resulting integrated quantity $\Theta^s$ will be a
functional $s$-form on the space of configurations $\cF$:
\begin{equation}
\Theta^s[z] = \int_\Sigma \theta^{n,s}.
\end{equation}

When $\delta$ acts on a functional $F$, we obtain:
\begin{equation}
\delta F[z] = \int_\Sigma \delta f d^nx = \int_\Sigma \delta z^A
\frac{\delta f}{\delta z^A} d^nx + \oint_{\d \Sigma} I^n(f d^nx),
\end{equation}
where $\frac{\delta}{\delta z^A}$ is the Euler-Lagrange derivative and
$I^n(f d^nx)$ denotes the $(n-1,1)$-form obtained by integration by
parts. In a similar way, we have
\begin{equation}
\delta_Q F[z] = \int_\Sigma Q^A
\frac{\delta f}{\delta z^A} d^nx + \oint_{\d \Sigma} I_Q^n(f d^nx),
\end{equation}
where $I^n_Q(f d^nx)$ is a $(n-1,0)$-form given by
\begin{equation}
\label{eq:homotopQ}
I^n_Q(f d^nx) = \iota_Q I^n(f d^nx).
\end{equation}

\subsection{Boundary Conditions}
\label{sec:boundary-conditions}

In the description of the differential structure we did not take into
account the boundary conditions. They will impose restrictions on both
$\delta$ and the allowed evolutionary vector fields.

The boundary conditions $\chi^\mu\vert_{\d \Sigma} = 0$ are valid for
all allowed field configurations: they must be preserved by
$\delta$. The exterior derivative $\delta$ satisfies
\begin{equation}
\delta \chi^\mu\vert_{\d \Sigma} = 0.
\end{equation}
In a similar way, an allowed evolutionary vector field must transform
allowed configurations into allowed configurations:
\begin{equation}
\delta_Q \chi^\mu\vert_{\d \Sigma} = 0.
\end{equation}

We then have the two following important results:
\begin{theorem}
For any $(p,q)$-form  $\theta^{n,s}$ such that we have
\begin{equation}
\left. \theta^{p,q} \right\vert_{\d \Sigma} = 0
\end{equation}
for all allowed values of $z^A$ and $\delta z^A$, we have 
\begin{gather}
\left. \delta \theta^{p,q} \right\vert_{\d \Sigma} =
0, \quad \left. \iota_Q \theta^{p,q} \right\vert_{\d \Sigma} = 0, \\ \left. \delta_Q \theta^{p,q} \right\vert_{\d \Sigma}= 0,
\end{gather}
for all allowed evolutionary vector fields $Q^A\frac{\d}{\d z^A}$.
\end{theorem}
\begin{corollary}
The set of evolutionary vector fields preserving the boundary conditions forms
an algebra.
\end{corollary}

\bibliography{../../../Docear/_data/14C4C119A9EB7LK4NIEX6AWNRSXHJH2EOCOG/default_files/Physics}

\providecommand{\href}[2]{#2}\begingroup\raggedright\begin{thebibliography}{10}

\bibitem{Goldstein2001}
J.~S. H.~Goldstein, C.~Goole, {\em Classical Mechanics, 3rd edition}.
\newblock Addison Wesley, 2001.

\bibitem{Gibbons1977}
G.~W. Gibbons and S.~W. Hawking, ``Action integrals and partition functions in
  quantum gravity,'' {\em Physical Review D} {\bf 15} (1977), no.~10, 2752.

\bibitem{Arnowitt2008}
R.~Arnowitt, S.~Deser, and C.~W. Misner, ``Republication of: The dynamics of
  general relativity,'' {\em General Relativity and Gravitation} {\bf 40}
  (2008), no.~9, 1997--2027.

\bibitem{Regge1974}
T.~Regge and C.~Teitelboim, ``Role of surface integrals in the hamiltonian
  formulation of general relativity,'' {\em Annals of Physics} {\bf 88} (1974),
  no.~1, 286--318.

\bibitem{Benguria1977}
R.~Benguria, P.~Cordero, and C.~Teitelboim, ``Aspects of the hamiltonian
  dynamics of interacting gravitational gauge and higgs fields with
  applications to spherical symmetry,'' {\em Nuclear Physics B} {\bf 122}
  (1977), no.~1, 61--99.

\bibitem{Henneaux1985}
M.~Henneaux and C.~Teitelboim, ``Asymptotically anti-de sitter spaces,'' {\em
  Communications in Mathematical Physics} {\bf 98} (1985), no.~3, 391--424.

\bibitem{Brown1986}
J.~D. Brown and M.~Henneaux, ``Central charges in the canonical realization of
  asymptotic symmetries: an example from three dimensional gravity,'' {\em
  Communications in Mathematical Physics} {\bf 104} (1986), no.~2, 207--226.

\bibitem{Brown1986a}
J.~D. Brown and M.~Henneaux, ``On the poisson brackets of differentiable
  generators in classical field theory,'' {\em Journal of mathematical physics}
  {\bf 27} (1986), no.~2, 489--491.

\bibitem{Barnich2002}
G.~Barnich and F.~Brandt, ``Covariant theory of asymptotic symmetries,
  conservation laws and central charges,'' {\em Nuclear Physics B} {\bf 633}
  (2002), no.~1, 3--82.

\bibitem{Barnich2003}
G.~Barnich, ``Boundary charges in gauge theories: Using stokes theorem in the
  bulk,'' {\em Classical and quantum gravity} {\bf 20} (2003), no.~16, 3685.

\bibitem{Barnich2008}
G.~Barnich and G.~Compere, ``Surface charge algebra in gauge theories and
  thermodynamic integrability,'' {\em Journal of Mathematical Physics} {\bf 49}
  (2008), no.~4, 042901.

\bibitem{Campoleoni2010}
A.~Campoleoni, S.~Fredenhagen, S.~Pfenninger, and S.~Theisen, ``Asymptotic
  symmetries of three-dimensional gravity coupled to higher-spin fields,'' {\em
  Journal of High Energy Physics} {\bf 2010} (2010), no.~11, 1--36.

\bibitem{Henneaux2010}
M.~Henneaux and S.-J. Rey, ``Nonlinear w∞ as asymptotic symmetry of
  three-dimensional higher spin ads gravity,'' {\em Journal of High Energy
  Physics} {\bf 2010} (2010), no.~12, 1--20.

\bibitem{Gaberdiel2011}
M.~R. Gaberdiel and T.~Hartman, ``Symmetries of holographic minimal models,''
  {\em Journal of High Energy Physics} {\bf 2011} (2011), no.~5, 1--26.

\bibitem{Gaberdiel2012}
M.~R. Gaberdiel and R.~Gopakumar, ``Minimal model holography,'' {\em arXiv
  preprint arXiv:1207.6697} (2012).

\bibitem{Moore1989}
G.~Moore and N.~Seiberg, ``Taming the conformal zoo,'' {\em Physics Letters B}
  {\bf 220} (1989), no.~3, 422--430.

\bibitem{Elitzur1989}
S.~Elitzur, G.~Moore, A.~Schwimmer, and N.~Seiberg, ``Remarks on the canonical
  quantization of the chern-simons-witten theory,'' {\em Nuclear Physics B}
  {\bf 326} (1989), no.~1, 108--134.

\bibitem{Lewis1986}
D.~Lewis, J.~Marsden, R.~Montgomery, and T.~Ratiu, ``The hamiltonian structure
  for dynamic free boundary problems,'' {\em Physica D: Nonlinear Phenomena}
  {\bf 18} (1986), no.~1, 391--404.

\bibitem{Soloviev1993}
V.~O. Soloviev, ``Boundary values as hamiltonian variables. i. new poisson
  brackets,'' {\em Journal of mathematical physics} {\bf 34} (1993), no.~12,
  5747--5769.

\bibitem{Soloviev1995}
V.~O. Soloviev, ``Boundary values as hamiltonian variables. ii. graded
  structures,'' {\em Arxiv preprint q-alg/9501017} (1995).

\bibitem{Arnold1989}
V.~Arnold, {\em Mathematical Methods of Classical Mechanics}.
\newblock Springer, 1989.

\bibitem{Henneaux1992}
M.~Henneaux and C.~Teitelboim, {\em Quantization of gauge systems}.
\newblock Princeton University Press, 1992.

\bibitem{Brown1993}
J.~D. Brown and J.~W. York~Jr, ``Quasilocal energy and conserved charges
  derived from the gravitational action,'' {\em Physical Review D} {\bf 47}
  (1993), no.~4, 1407.

\bibitem{Henneaux2004}
M.~Henneaux, C.~Martinez, R.~Troncoso, and J.~Zanelli, ``Asymptotically
  anti--de sitter spacetimes and scalar fields with a logarithmic branch,''
  {\em Physical Review D} {\bf 70} (2004), no.~4, 044034.

\bibitem{Henneaux2007}
M.~Henneaux, C.~Mart{\'\i}nez, R.~Troncoso, and J.~Zanelli, ``Asymptotic
  behavior and hamiltonian analysis of anti-de sitter gravity coupled to scalar
  fields,'' {\em Annals of Physics} {\bf 322} (2007), no.~4, 824--848.

\bibitem{Coussaert1995}
O.~Coussaert, M.~Henneaux, and P.~van Driel, ``The asymptotic dynamics of
  three-dimensional einstein gravity with a negative cosmological constant,''
  {\em Classical and Quantum Gravity} {\bf 12} (1995), no.~12, 2961.

\bibitem{Henneaux2000}
M.~Henneaux, L.~Maoz, and A.~Schwimmer, ``Asymptotic dynamics and asymptotic
  symmetries of three-dimensional extended ads supergravity,'' {\em Annals of
  Physics} {\bf 282} (2000), no.~1, 31--66.

\bibitem{Horowitz1989}
G.~T. Horowitz, ``Exactly soluble diffeomorphism invariant theories,'' {\em
  Communications in Mathematical Physics} {\bf 125} (1989), no.~3, 417--437.

\bibitem{Cattaneo}
A.~S. Cattaneo, P.~Cotta-Ramusino, J.~Frohlich, and M.~Martellini,
  ``Topological bf theories in three-dimensions and four-dimensions,'' {\em
  Journal of Mathematical Physics} {\bf 36} 6137.

\bibitem{Baez1996}
J.~C. Baez, ``Four-dimensional bf theory as a topological quantum field
  theory,'' {\em Letters in Mathematical Physics} {\bf 38} (1996), no.~2,
  129--143.

\end{thebibliography}\endgroup

%\bibliography{/Users/ctroessa/Physics/Bibliography/Papers/biblio.bib}
\end{document}